\pdfoutput=1
\newif\ifdraft \draftfalse
\newif\iffull \fulltrue

\makeatletter \@input{tex.flags} \makeatother

\iffull
\documentclass[final,leqno,onefignum,onetabnum]{siamart}
\else
\documentclass{sig-alternate}
\fi

\newenvironment{mathdisplayfull}{\iffull\begin{displaymath}\else $\fi}{\iffull
  \end{displaymath}\else $\ignorespaces\fi}

\newcommand{\shortbreak}{{\iffull \else \\ \fi}}

\newcommand{\longbreak}{{\iffull \\ \else\fi}}
\newcommand{\longquad}{{\iffull \qquad \else \;\;  \fi}}
\newcommand{\fullfrac}[2]{\iffull \frac{#1}{#2} \else #1 / #2 \fi}
\newcommand{\thelongref}[1]{\iffull\cref{#1}\else the full version\fi}

\usepackage{algorithm}
\usepackage[noend]{algorithmic}
\usepackage{amsmath, amssymb}
\usepackage{mathtools}
\usepackage{color}
\usepackage{xcolor}
\usepackage{xspace}
\usepackage{multicol}
\usepackage{mathpartir}

\iffull
\usepackage{natbib}
\else
\usepackage[numbers,sort&compress]{natbib}
\let\proof\relax
\let\endproof\relax
\usepackage{amsthm}
\usepackage{verbatim}
\usepackage{times}
\usepackage{hyperref}
\usepackage[capitalize]{cleveref}
\fi

\usepackage{thmtools, thm-restate}

\definecolor{DarkGreen}{rgb}{0.1,0.5,0.1}
\definecolor{DarkRed}{rgb}{0.5,0.1,0.1}
\definecolor{DarkBlue}{rgb}{0.1,0.1,0.5}
\hypersetup{
    unicode=false,          % non-Latin characters in Acrobat's bookmarks
    pdftoolbar=true,        % show Acrobat toolbar?
    pdfmenubar=true,        % show Acrobat menu?
    pdffitwindow=false,      % page fit to window when opened
    pdftitle={},    % title
    pdfauthor={}
    pdfsubject={},   % subject of the document
    pdfnewwindow=true,      % links in new window
    pdfkeywords={keywords}, % list of keywords
    colorlinks=true,       % false: boxed links; true: colored links
    linkcolor=DarkRed,          % color of internal links
    citecolor=DarkGreen,        % color of links to bibliography
    filecolor=DarkRed,      % color of file links
    urlcolor=DarkBlue,          % color of external links
}

\newcommand{\ar}[1]{\ifdraft \textcolor{brown}{[Aaron: #1]}\fi}

\newcommand\RR{\mathbb{R}}

\newcommand\cM{\mathcal{M}}

\newcommand\cR{\mathcal{R}}

\newcommand\cD{\mathcal{D}}

\newcommand{\bin}{\text{Bin}}

\DeclareMathOperator{\polylog}{polylog}

\renewcommand{\tilde}{\widetilde}

\newcommand{\prob}[1]{\Pr\left[#1\right]}

\newcommand{\eps}{\varepsilon}
\def\epsilon{\varepsilon}

\DeclareMathOperator{\Lap}{Lap}
\DeclareMathOperator{\OPT}{OPT}
\renewcommand{\hat}{\widehat}
\renewcommand{\bar}{\overline}

\DeclareMathOperator*{\argmax}{\mathrm{argmax}}

\iffull
\newtheorem{remark}[theorem]{Remark}
\else
\newtheorem*{theorem*}{Theorem}
\declaretheorem{theorem}
\declaretheorem{lemma}

\declaretheorem{remark}

\declaretheorem{definition}
\fi

\iffull\else
\fi

\newcommand{\counter}{{{\bf Counter}}\xspace}
\newcommand{\pmatch}{{{\bf PMatch}}\xspace}
\newcommand{\palloc}{{{\bf PAlloc}}\xspace}

\title{Private Matchings and Allocations\iffull\thanks{%
    An extended abstract of this paper appeared in STOC 2014
    \citep{HHRRW14}.}\fi}

\iffull
\author{
  Justin Hsu\thanks{
    Department of Computer and Information Science, University of Pennsylvania.
    Supported in part by NSF Grant CNS-1065060.
    Email: {\tt justhsu@cis.upenn.edu}.}
  \and Zhiyi Huang\thanks{
    Department of Computer Science, University of Hong Kong.
    Email: {\tt zhiyi@cs.hku.hk}.
    This work is done in part while the author was a student at the University of Pennsylvania.}
  \and Aaron Roth\thanks{
    Department of Computer and Information Science, University of Pennsylvania.
    Supported in part by an NSF CAREER award, NSF Grants CCF-1101389 and
    CNS-1065060, and a Google Focused Research Award.
    Email: {\tt aaroth@cis.upenn.edu}.}
  \and
  \linebreak
  Tim Roughgarden\thanks{
    Department of Computer Science, Stanford University,
    462 Gates Building, 353 Serra Mall, Stanford, CA 94305.
    This research was supported in part by NSF Awards CCF-1016885 and
    CCF-1215965 and an ONR PECASE Award.
    Email: {\tt tim@cs.stanford.edu}.
    This work done in part while visiting University of Pennsylvania.}
  \and Zhiwei Steven Wu\thanks{
    Department of Computer and Information Science, University of Pennsylvania.
    Supported in part by NSF Grant CCF-1101389.
    Email: {\tt wuzhiwei@cis.upenn.edu}}}
\else
\numberofauthors{5} %  in this sample file, there are a *total*
\author{
\alignauthor
Justin Hsu\titlenote{
  Supported in part by NSF Grant CNS-1065060.}\\
\affaddr{University of Pennsylvania}
\alignauthor
Zhiyi Huang\\
\affaddr{Stanford University \qquad University of Hong Kong}
\alignauthor
Aaron Roth\titlenote{
    Supported in part by an NSF CAREER award, NSF Grants CCF-1101389 and
    CNS-1065060, and a Google Focused Research Award.}\\
\affaddr{University of Pennsylvania}
\and
\alignauthor
Tim Roughgarden\titlenote{
  This research was supported in part by NSF Awards CCF-1016885 and CCF-1215965
  and an ONR PECASE Award.}\\
\affaddr{Stanford University}
\alignauthor
Zhiwei Steven Wu\titlenote{
  Supported in part by NSF Grant CCF-1101389.}\\
\affaddr{University of Pennsylvania}
}
\newfont{\mycrnotice}{ptmr8t at 7pt}
\newfont{\myconfname}{ptmri8t at 7pt}

\fi

\begin{document}

\iffull\else
\conferenceinfo{STOC '14,}{May 31--June 03 2014, New York, NY, USA\\
{\mycrnotice{Copyright is held by the owner/author(s). Publication rights
    licensed to ACM.}}}
\copyrightetc{ACM \the\acmcopyr}
\crdata{978-1-4503-2710-7/14/05\ ...\$15.00.\\
  \url{http://dx.doi.org/10.1145/2591796.2591826}}
\fi

\maketitle
\iffull
\slugger{sicomp}{xxxx}{xx}{x}{x--x}%slugger should be set to mms, siap, sicomp, sicon, sidma, sima, simax, sinum, siopt, sisc, or sirev
\fi

\begin{abstract}
  % \jh{Possible more concise titles: ``Matching and Allocating, Privately''.
  %   ``Private Matchings and Allocations''. ``Matchings and Allocations,
  %   Privately''.}
  % \sw{``Privacy Preserving Matching and Allocations''}
We consider a private variant of the classical \emph{allocation problem}: given
$k$ goods and $n$ agents with private valuation functions over
bundles of goods, how can we allocate goods to agents to maximize
social welfare? An important special case is when agents desire at most one
good, and specifies their (private) value for each good: in this case, the problem
is exactly the maximum-weight matching problem in a bipartite graph.

Private matching and allocation problems have not been considered in the
differential privacy literature for a good reason: they are plainly
impossible to solve under differential privacy.  Informally, the allocation must
match agents to their preferred goods in order to maximize social welfare, but this
preference is exactly what agents wish to hide! Therefore, we consider the
problem under the relaxed constraint of \emph{joint differential privacy}: for
any agent $i$, no coalition of agents excluding $i$ should be able to learn
about the valuation function of agent $i$. In this setting, the full
allocation is no longer published---instead, each agent is told what good to
receive. We first show that if there are several identical copies of each good,
it is possible to
efficiently and accurately solve the matching problem while
guaranteeing joint differential privacy.  We then consider the
more general allocation problem where bidder valuations satisfy the \emph{gross
  substitutes} condition.  Finally, we prove that the allocation problem cannot
be solved to non-trivial accuracy under joint differential privacy without
requiring multiple copies of each type of good.
\end{abstract}

\iffull

\pagestyle{myheadings}
\thispagestyle{plain}
\markboth{Private Matchings and Allocations}
{J. Hsu, Z. Huang, A. Roth, T. Roughgarden, Z.S. Wu}
\else
\fi

\section{Introduction}
In the classic maximum-weight matching problem in bipartite graphs,
there are $k$ goods $j \in \{1,\ldots, k\}$ and $n$ buyers $i \in
\{1,\ldots,n\}$. Each buyer $i$ has a value $v_{ij} \in [0,1]$ for each good
$j$, and the goal is to find a matching $\mu$ between goods and buyers which
maximizes the social welfare $\mathrm{SW} = \sum_{i=1}^n v_{i,\mu(i)}$. When
the buyers' values are sensitive information,\footnote{%
  For instance, the goods might be related to the treatment of disease, or might
  be indicative of a particular business strategy, or might be embarrassing in
  nature.}
it is natural to ask for a matching that hides the reported values of each of
the players.

It is not hard to see that this goal is impossible under the standard notion of
differential privacy, which requires that the allocation must be insensitive to
the reported valuations of each player. We formalize this observation in
\cref{sec:lowerbounds}, but the intuition is simple. Consider the case with two
types of goods with $n$ identical copies each, and
suppose that each buyer has a private preference for one of the two types: value
$1$ for the good that he likes, and value $0$ for the other good. There is no
contention since the supply of each good is larger than the total number of
buyers, so any allocation achieving social welfare $\OPT - \alpha n$ can be used
to reconstruct a $(1-\alpha)$ fraction of the preferences; this is plainly
impossible for non-trivial values of $\alpha$ under differential privacy.

In light of this obstacle, is there any hope for privately solving maximum-weight
matching problems?  In this paper, we show that the answer is \emph{yes}: it is
possible to solve matching problems (and more general allocation problems) to
high accuracy assuming a small number of identical copies of each good, while
still satisfying an extremely strong variant of differential privacy. We observe
that the matching problem has the following two features:
\begin{enumerate}
  \item Both the input and solution are naturally partitioned amongst the same
    $n$ people: each buyer $i$ receives the item $\mu(i)$ they are matched to in
    the solution.
  \item The problem is not solvable privately because the item given to each
    buyer must reflect their own private data.
\end{enumerate}
By utilizing these two features, we show that the matching problem can be
accurately solved under the constraint of \emph{joint
  differential privacy} \citep{kearns-largegame}.  Informally speaking, this
requires that for every buyer $i$, the joint distribution on items $\mu(j)$ for
$j \neq i$ must be differentially private in the reported valuation of buyer
$i$. As a consequence, buyer $i$'s privacy is protected even if {\em all} other
buyers collude, potentially sharing the identities of the items they
receive. As long as buyer $i$ does not reveal their own item, $i$'s privacy is
protected.

We then show that our techniques generalize beyond the max-matching
problem to the more general \emph{allocation} problem. Here, each
buyer $i$ has a valuation function defined over subsets of goods
$v_i:2^{[k]}\rightarrow [0,1]$ from some class of valuations, and the goal
is to find a partition of the goods $S_1,\ldots,S_n$ maximizing social welfare;
note that the maximum-weight matching problem is the special case when agents are
\emph{unit demand}, i.e., only want bundles of size $1$. More specifically, we
consider buyers with
\emph{gross substitutes} valuations. This is an economically
meaningful class of valuation functions that is a strict subclass of
submodular functions and are the most general class of
valuations for which our techniques apply.

\subsection{Our Techniques and Results}
Our approach makes a novel connection between \emph{market clearing prices} and
differential privacy. Prices have long been considered as a low-information way
to coordinate markets; our paper formalizes this intuition in the context of
differentially private allocation.  Specifically, we will use \emph{Walrasian
  equilibrium prices}: prices under which each buyer is simultaneously able to
buy a most preferred bundle of goods, and no good is over-demanded.  Although
the allocation itself cannot be computed under standard differential privacy, we
show how to differentially privately compute the Walrasian equilibrium prices while
coordinating a high welfare allocation under joint differential privacy.

We start from the classic analysis of \citet{job-matching}, who show
how to use \emph{ascending price} auctions to compute Walrasian equilibrium
prices. In the classical ascending price auction, each good begins with a
price of $0$ and each agent is initially unmatched to any good.  Unmatched
agents $i$ take turns bidding on the good $j^*$ that maximizes their utility at
the current prices: i.e., $j^* \in \arg\max(v_{ij} - p_j)$. When a bidder bids
on a good $j^*$, they become the new high bidder and the price of $j^*$ is
incremented.  Bidders are tentatively matched to a good as long as
they are the high bidder. The auction continues until there are no unmatched
bidders who prefer to be matched at the current
prices. The algorithm converges because each bid increases the
the prices, which are bounded by some finite value.\footnote{%
  Bidders do not bid on goods for which they have negative utility; in our case,
  $v_{ij} \in [0,1]$.}
Moreover, every bidder ends up matched to their most preferred
good given the prices.  Finally, by the \emph{first welfare theorem} of Walrasian
equilibria, any matching that corresponds to equilibrium prices maximizes
social welfare.  We emphasize that this final implication
is key: ``prices'' play no role in our problem description, nor do we ever
actually charge ``prices'' to the agents---the prices are purely a
device to coordinate the matching.

We give an approximate, private version of Kelso and Crawford's algorithm based on several
observations. First, in order to implement this algorithm, it is sufficient to
maintain the sequence of prices of the goods privately: given a record of the
price trajectory, each agent can figure out what good they are
matched to. Second, in order to privately maintain the prices, it suffices to
maintain a private count of the number of bids each good has received over the
course of the auction; we can accomplish this task using private counters due to
\citet{DNPR10,chan-counter}.  Finally, it is possible to halt the algorithm
early without significantly harming the quality of the final matching. By doing
so, we reduce the number of bids from each bidder, enabling us to bound the
sensitivity of the bid counters, reducing the amount of noise needed for
privacy.
% \footnote{%
%   These papers show that a single sensitivity-1 counter can be implemented
%   privately. We need a stronger version of this theorem in our work: a
%   collection of counters operating on adaptively chosen streams, whose
%   sensitivity can be bounded jointly, should be private with a privacy parameter
%   that depends only on the joint-sensitivity of the streams and is independent
%   of the number of counters. This stronger statement also holds, as we show in
%   our privacy analysis.}
The result is an algorithm that converges to a matching together with prices
that form an approximate Walrasian equilibrium. We complete our analysis by
proving an approximate version of the first welfare theorem, which shows that
the matching has high weight.

The algorithm of \citet{job-matching} extends to the general allocation problem
when players have gross substitute preferences, and our private algorithm does
as well. We note that this class of preferences is the natural limit of our
approach, which makes crucial use of equilibrium prices as a coordinating
device: in general, when agents have valuations over bundles of goods that do
not satisfy the gross substitutes condition, Walrasian equilibrium prices may not
exist.

We first state our main result informally in the special case of max-matchings,
which we prove in \cref{sec:matchings}. We prove our more general theorem for
allocation problems with gross substitutes preferences in \cref{sec:extensions}.
Here, privacy is protected with respect to a single agent $i$ changing their
valuations $v_{ij}$ for possibly \emph{all} goods $j$.

\begin{theorem*}[Informal]
  Suppose there are $n$ agents and $k$ types of goods, with each with $s$
  identical copies.  There is a computationally efficient $\epsilon$-joint
  differentially private algorithm which computes a matching of weight
  $\mathrm{OPT}-\alpha n$ as long as
  \[
    s \geq O\left(\frac{1}{\alpha^3 \epsilon} \cdot \polylog \left(n, k,
        \frac{1}{\alpha} \right) \right).
  \]
  For certain parameter ranges, the welfare guarantee can be improved to
  $(1-\alpha)\OPT$.
\end{theorem*}

Our algorithms actually work in a privacy model that is stronger than joint
differential privacy, called the {\em
  billboard model}. We can view the algorithm as a mechanism that posts the prices publicly on a {\em billboard}
as a differentially private signal such that every player can deduce what object
they should be matched to just from their own private information and the contents
of the billboard. As we show, algorithms in the billboard model automatically
satisfy joint differential privacy.

Furthermore, we view implementations in the billboard model as preferable to
arbitrary jointly differentially private implementations.
Algorithms in the billboard model only need the ability to publish sanitized
messages to all players, and do not need a secure channel to communicate the
mechanisms' output to each player (though of course, there still needs to be a
secure channel from the player to the mechanism). The previous work by \citet{MM09} and
some of the results by \citet{GLMRT10} can be viewed as existing examples of
algorithms in the billboard model.

In \cref{sec:lowerbounds}, we complement our positive results with lower bounds
showing that our results are qualitatively tight.  Not only is the problem
impossible to solve under the standard differential privacy, assuming multiple
copies of each good is also necessary to get any non-trivial solution even under {\em
  joint} differential privacy.
\begin{theorem*}[Informal]
No joint differentially private algorithm can compute matchings of weight
greater than $\mathrm{OPT} - \alpha n$ on instances in which there are $n$
agents and $s$ copies of each good, when
\begin{mathdisplayfull}
  s \leq O\left(\fullfrac{1}{\sqrt{\alpha}}\right).
\end{mathdisplayfull}
\end{theorem*}
In particular, no algorithm can compute matchings of weight $\mathrm{OPT} -
o(n)$ on instances for which the supply $s = O(1)$. In addition, we show that
when goods have supply only $s = O(1)$, it is not even possible to compute the
equilibrium prices privately under standard differential privacy.
Our lower bounds are all reductions to database reconstruction attacks. Our
technique for proving this lower bound may be of general interest, as the
construction may be useful for other lower bounds for joint differential
privacy.

% \ar{Verify that this is the lower bound we prove -- its currently stated in a
%   confusing way in the lower bounds section.}
% \jh{Looks good to me. Fixed.}

\subsection{Related Work}
Differential privacy, first defined by \citet{DMNS06}, has become a standard
privacy solution concept in the theoretical computer science literature.
There is far too much work to survey comprehensively; for a textbook
introduction, see \citet{DR13}.

The privacy of our algorithms relies on work by \citet{DNPR10} and
\citet{chan-counter}, who show how to release a running count of a stream of
bits under \emph{continual observation}---i.e., report the count as the stream
is revealed, provide high accuracy at every point in time, and keep the
transcript differentially private.

Beginning with \citet{DN03}, much work in differential privacy has focused on
answering numeric valued queries on a private dataset (e.g.,
\citet{DMNS06,BLR08,HR10}, among many others).  In contrast, work on private
combinatorial optimization problems has been sporadic (e.g.,
\citet{NRS07,GLMRT10}). Part of the challenge is that many combinatorial
optimization problems, including the allocation problems we consider in this
paper, are impossible to solve under differential privacy. To sidestep this
problem, we employ the solution concept of {\em joint differential privacy}.
First formalized by \citet{kearns-largegame}, similar ideas are present in the
vertex and set-cover algorithms of \citet{GLMRT10}, the private recommendation
system of \citet{MM09}, and the analyst private data analysis algorithms of
\citet{DNV12,HRU13}.
% \jh{We informally define joint privacy at least 3-4 times. Maybe we can
%   consolidate.}
% \ar{Ya, ya... Removed the informal description}

Our algorithm is inspired by \citet{job-matching}, who
study the problem of matching {\em firms} to \emph{workers} when the firms have
preferences that satisfy the \emph{gross substitutes} condition. They give an
algorithm based on simulating simultaneous ascending auctions that converge to
\emph{Walrasian equilibrium prices} and a corresponding matching.  In
some respect, this approach does not generalize to more general valuations:
\citet{GS99} show that gross substitutes preferences are precisely the set of
preferences for which Walrasian equilibrium prices are guaranteed to exist.
% \jh{Was a typo in this last sentence: ``our approach extends to further:''.
%   Changed ``to'' to ``no'', please fix if that's not the intended meaning.}
% \jh{Changed ``extends no further'' to ``is complete''. Again, change if wrong
%   meaning.}
% \ar{both changes look good!}

While our algorithm achieves good approximation to the optimal
welfare at the expense of certain incentive properties, our work is closely
related to recent work on privately computing various kinds of equilibrium in
games (e.g., correlated equilibrium \citep{kearns-largegame}, Nash equilibrium
\citep{RR13}, and minmax equilibrium \citep{HRU13}). These works belong to a
growing literature studying the interface of game theory and differential
privacy; \citet{PR13} provide a recent survey.

\section{Preliminaries}
\subsection{The Allocation Problem}
We consider allocation problems defined by a set of goods $G$, and
a set of $n$ agents $[n]$. Each agent $i \in [n]$ has a \emph{valuation
function} $v_i:2^G\rightarrow [0,1]$ mapping bundles of goods to values.
A \emph{feasible allocation} is a collection of sets $S_1,\ldots,S_n \subseteq G$
such that $S_i \cap S_j = \emptyset$ for each $i \neq j$: i.e., a partition of
goods among the agents.  The
\emph{social welfare} of an allocation $S_1,\ldots,S_n$ is
$\sum_{i=1}^n v_i(S_i)$, the sum of the agent's valuations for the allocation;
we are interested in finding allocations which maximize this quantity.  Given an
instance of an allocation problem, we write $\mathrm{OPT} =
\max_{S_1,\ldots,S_n}\sum_{i=1}^n v_i(S_i)$ to denote the social welfare of the
optimal feasible allocation.

A particularly simple valuation function is a \emph{unit demand valuation},
where bidders demand at most one item. Such valuation functions take the form
$v_i(S) = \max_{j \in S} v_i(\{j\})$ and can be specified by numbers $v_{i,j} =
v_i (\{j\})\in [0,1]$, which represent the value that bidder $i$ places on good
$j$. When bidders have unit demand valuations, the allocation problem
corresponds to computing a maximum weight matching in a bipartite graph.

Our results will also hold for {\em gross substitute valuations},
which include unit demand valuations as a special case.  Informally,
for gross substitute valuations, any set of goods $S'$ that are in a
most-demanded bundle at some set of prices $p$ remain in a
most-demanded bundle if the prices of \emph{other} goods are raised,
keeping the prices of goods in $S'$ fixed. Gross substitute valuations
are a standard class of valuation functions: they are a strict
subclass of submodular functions, and they are precisely the valuation
functions with Walrasian equilibria in markets with indivisible goods
\citep{GS99}. Two other simple examples of gross substitute
valuations are (1) {\em additive functions}, which takes the form
$v(S) = \sum_{j\in S} v(\{j\})$ and (2) {\em symmetric submodular
  functions}, such that $v(S) = f(|S|)$ for some monotone concave
function $f\colon \RR_+ \rightarrow \RR_+$.

To give the formal definition, we will need some notation. Given a
vector of prices $\{p_g\}_{g \in G}$, the (quasi-linear) \emph{utility} that
player $i$ has for a bundle of goods $S_i$ is defined to be $u_i(S_i, p) =
v_i(S_i) - \sum_{j \in S_i} p_j$.\footnote{%
  This is a natural definition of utility if agents must pay for the bundles
  they buy at the given prices.  In this paper we are concerned with the purely
  algorithmic allocation problem, so our algorithm will not actually charge
  prices. However, prices will be a convenient abstraction throughout our
  work.}
Given a vector of prices $p$, for each agent $i$ we can define the set of
\emph{most demanded bundles}: $\omega(p) = \arg\max_{S \subseteq G} u_i(S, p)$.
Given two price vectors $p, p'$, we write $p \preceq p'$ if $p_g \leq
p'_g$ for all $g$.
\begin{definition} \label{def-gs}
  A valuation function $v_i:2^G\rightarrow [0,1]$ satisfies the {\em gross
  substitutes condition} if for every two price vectors $p \preceq p'$ and
  for every bundle $S \in \omega(p)$, if $S' \subseteq S$ satisfies $p'_g
  = p_g$ for every $g \in S'$, then there is a bundle $S^* \in \omega(p')$ with $S'
  \subseteq S^*$.
\end{definition}
Finally, we will typically consider markets with multiple copies of each type of good. Two
goods $g_1,g_2 \in G$ are \emph{identical} if for every bidder $i$ and for every
bundle $S \subseteq G$, $v_i(S \cup \{g_1\}) = v_i(S \cup \{g_2\})$: i.e., the two
goods are indistinguishable according to every valuation function.  Formally, we
say that a set of goods $G$ consists of $k$ {\em types} of goods with $s$ {\em
  supply} if there are $k$ representative goods $g_1,\ldots,g_k \in G$ such that
every good $g' \in G$ is identical to one of $g_1,\ldots,g_k$, and for each
representative good $g_i$, there are $s$ goods identical to $g_i$ in $G$. For
simplicity of presentation we will assume that the supply of
each good is the same, but this is not necessary; all of our results continue
to hold when the supply $s$ denotes the \emph{minimum} supply of any type of
good. 

\subsection{Differential Privacy Preliminaries}
Although it is impossible to solve the allocation problem under standard
differential privacy (see \cref{sec:lowerbounds}), standard differential privacy
plays an essential role in our analysis. We will introduce this concept before
seeing its generalization, joint differential privacy.

Suppose agents have valuation functions $v_i$ from a class of functions $C$. A
database $D \in C^n$ is a vector of valuation functions, one for each of the $n$
bidders.  Two databases $D, D'$ are $i$-\emph{neighbors} if they differ in only
their $i$'th index: that is, if $D_j = D'_j$ for all $j \neq i$. If two
databases $D, D'$ are $i$-neighbors for some $i$, we say that they are
\emph{neighboring databases}. We will be interested in randomized algorithms
that take a database as input, and output an element from some range $\cR$. Our
final mechanisms will output sets of $n$ bundles (so $\cR = (2^G)^n$), but
intermediate components of our algorithms will have different ranges.
% \jh{We only use $\epsilon, \delta$-privacy when proving lower bounds and the
%   billboard lemma. Are the slightly more general results worth it?}
\begin{definition}[\citet{DMNS06}]
  An algorithm $\cM:C^n\rightarrow \cR$ is \\
  $(\epsilon,\delta)$-\emph{differentially
  private} if for every pair of neighboring databases $D, D' \in C^n$ and for
  every set of subset of outputs $S \subseteq \cR$,
\[
  \Pr[\cM(D) \in S] \leq e^\epsilon \Pr[\cM(D') \in S] + \delta.
\]
  If $\delta = 0$, we say that $\cM$ is {\em $\epsilon$-differentially private}.
\end{definition}
When the range of a mechanism is also a vector with $n$ components (e.g., $\cR =
(2^G)^n$), we can define \emph{joint differential privacy}: this requires that
simultaneously for all $i$, the \emph{joint} distribution on outputs given to
players $j \neq i$ is differentially private in the input of agent $i$. Given a
vector $x = (x_1,\ldots,x_n)$, we write $x_{-i} =
(x_1,\ldots,x_{i-1},x_{i+1},\ldots,x_n)$ to denote the vector of
length $n-1$ which contains all coordinates of $x$ except the $i$'th coordinate.
\begin{definition}[\citet{kearns-largegame}]
  An algorithm $\cM:C^n \rightarrow (2^G)^n$ is {\em $(\epsilon,\delta)$-joint
  differentially private} if for every $i$, for every pair of $i$-neighbors $D, D'
\in C^n$, and for every subset of outputs $S \subseteq (2^G)^{n-1}$,
\[
  \Pr[\cM(D)_{-i} \in S] \leq e^\epsilon \Pr[\cM(D')_{-i} \in S] + \delta.
\]
If $\delta = 0$, we say that $\cM$ is {\em $\epsilon$-joint differentially
  private}.
\end{definition}

Note that this is still an extremely strong definition that protects $i$ from
arbitrary coalitions of adversaries---it weakens the constraint of differential
privacy only in that the output given specifically to agent $i$ may be
sensitive in the input of agent $i$.

\subsection{Differentially Private Counters}
The central tool in our algorithm is the private streaming counter
proposed by \citet{chan-counter} and \citet{DNPR10}. Given a bit stream  $\sigma =
(\sigma_1, \ldots , \sigma_T)\in \{0,1\}^T$, a streaming counter $\cM(\sigma)$
releases an approximation to $c_\sigma(t) = \sum_{i=1}^t\sigma_i$ at every time
step $t$. The counters release accurate approximations to the running count at
every time step.

\begin{definition}
  A streaming counter $\cM$ is {\em $(\alpha, \beta)$-useful} if with
  probability at least $1 - \beta$, for each time $t \in [T]$,
  \[
    \left| \cM(\sigma)(t) - c_\sigma(t) \right| \leq \alpha.
  \]
\end{definition}

For the rest of this paper, let $\counter(\epsilon, T)$ denote the Binary
mechanism of \citet{chan-counter}, instantiated with parameters  $\epsilon$ and
$T$. The mechanism produces a monotonically increasing count, and satisfies the
following accuracy guarantee. Further details may be found in
\thelongref{counter-details}.

\begin{restatable}[\citet{chan-counter}]{theorem}{counteraccuracy}
  \label{counter-error}
  For $\beta > 0$, $\counter(\epsilon, T)$ is
  $\epsilon$-differentially private with respect to a single bit change in the
  stream, and $(\alpha, \beta)$-useful for
  \[
    \alpha = \frac{2\sqrt{2}}{\epsilon} \ln \left( 2/\beta\right)
    \log(T)^{5/2}.
   \]
\end{restatable}

% \ar{Removed privacy part of theorem, since we don't use this and reprove it as
%   we need it}
% \jh{That's reasonable, but then the $\epsilon$ looks a little out of place.
%   Should we just remove it?}
% \sw{We should specify the parameters of \counter: $(\epsilon, T)$, just  like
%   how we use them later; we might still want to provide some intuition about the
%   partial sums here, since we talk about them in privacy proof}
% \jh{I put it back. I think it's clearer this way, and now that we've moved
%   almost all the privacy details into the appendix, I think this is best.}
% \sw{We might want to point out that this is private w.r.t neighboring
%   streams, as we also removed the definition of streaming privacy. }
% \jh{I think this is handled.}

% \jh{It's too early to talk about this.}
% Our allocation algorithm applies \counter in two places:  we use a counter for
% each good to keep track the total number of bids it has received over time, and
% we use a special counter to count the number of unmatched bidders at the end of
% each round.

% We also require that the counters monotonically increase.  \citet{chan-counter}
% show that this can be accomplished by a post-processing step (which has no
% affect on privacy) with no loss in accuracy.

% \sw{Is ``asymptotic'' needed here? \citet{chan-counter} shows that the
% exact the same accuracy is preserved under this procedure.}
% \ar{Ya, not needed. I removed it.}

%%% Local Variables:
%%% mode: latex
%%% TeX-master: t
%%% End:

% \section{A Private Algorithm for Maximum-Weight Matching}
\section{Private Max-Weight Matching}
\label{sec:matchings}

In this section, we study the special case of unit demand valuations. Though our
later algorithm for gross substitutes valuations generalizes this case, we first
present our algorithm in this simpler setting to highlight the key features of
our approach.

Consider a matching market with $n$ bidders and $k$ different types of goods,
where each good has supply $s$ and bidder $i$ has valuation $v_{ij}\in [0,1]$
for good $j$. Some agents may not end up being matched to a good: to simplify
notation, we will say that unmatched agents are matched  to a special
dummy good $\bot$.

% \sw{I thought unmatched would mean $\mu(i) =$ empty, and $\perp$ is
%   actually some bidders favorite good given the prices.}
% \ar{Its true that we make the distinction in the -analysis- between bidders who
%   are unmatched because $\perp$ is their favorite good, and bidders who are
%   unmatched because we halted the algorithm early. However in the end we are
%   just talking about outputting a max-weight matching, and anyone ``matched to
%   $\perp$'' is just unmatched.}
% \sw{Sure. To be pedantic, we should change empty to $\perp$ when we
%   halt the algorithm}

To reach a maximum weight matching, we first aim to
privately compute prices $p\in [0,1]^k$ and an allocation of the goods
$\mu\colon [n]\rightarrow [k] \cup \{ \perp \}$ such that \emph{most} bidders
are  matched with their \emph{approximately} favorite goods \emph{given the
  prices} and each over-demanded good almost clears, where a
good is {\em over-demanded} if its price is strictly positive.\footnote{%
  This is the notion of approximate Walrasian equilibrium we will use.}
We will show that if we can achieve this intermediate goal, then in fact we have
computed an approximate maximum weight matching.

\begin{definition} \label{matching-eq}
  A price vector $p\in [0,1]^k$ and an assignment $\mu\colon
  [n]\rightarrow [k] \cup \{ \perp \}$ of bidders to goods
  is an {\em $(\alpha, \beta, \rho)$-approximate matching equilibrium} if:
  \begin{enumerate}
    \item all but a $\rho$ fraction of bidders $i$ are matched to an \shortbreak
      $\alpha$-approximate favorite good: i.e.,$v_{i \mu(i)} - p_{\mu(i)} \geq
      v_{ij} - p_j - \alpha$ for every good $j$,  for at least $(1-\rho)n$
      bidders $i$ (we call these bidders {\em satisfied});
    \item the number of bidders assigned to any type of good is below its
      supply; and
    \item each over-demanded good clears except for at most $\beta$ supply.
  \end{enumerate}
\end{definition}
% \sw{We should unify the uses of favorite vs most preferred; unmatched
%   vs. unsatisfied}

\subsection{Overview of the Algorithm}
Our algorithm takes the valuations as input, and outputs a trajectory of prices
that can be used by the agents to figure out what they are matched to.  For the
presentation, we will sometimes speak as if the bidders are performing some
action, but this actually means that our algorithm simulates the actions of the
bidders internally---the actual agents do not interact with our algorithm.

\cref{alg:matching} (\pmatch) is a variant of a {\em deferred acceptance}
algorithm first proposed and analyzed by \citet{job-matching}, which runs $k$
simultaneous ascending price auctions: one for each type of good. At any given
moment each type of good has a {\em proposal price} $p_j$. In a sequence of
rounds where the algorithm passes through each bidder once in some fixed,
publicly known order, unsatisfied bidders bid on a good that maximizes their
utility at the current prices: that is, a good $j$ that maximizes $v_{ij} -
p_j$.  (This is the $\mathbf{Propose}$ function.)

The $s$ most recent bidders for a type of good are tentatively matched to that
type of good; these are the current \emph{high bidders}. A bidder tentatively
matched to a good with supply $s$ becomes unmatched once the good
receives $s$ subsequent bids; we say this bidder has has been \emph{outbid}.
Every $s$ bids on a good increases its price by a fixed increment $\alpha$.
Bidders keep track of which good they are matched to, if
any, and determine whether they are currently matched or unmatched by
looking at a count of the number of bids received by the last good they bid on.

To implement this algorithm privately, we count the number of bids each good has
received using private counters. Unsatisfied bidders can infer the prices of all
goods based on the number of bids each has received, and from this information,
they determine their favorite good at the given prices.
Their bid is recorded by sending the bit 1 to the appropriate counter. (This is
the $\mathbf{Bid}$ function.) Matched bidders store the reading of the bid
counter on the good they are matched to at the time that they last bid (in the
variable $d_i$); when the counter ticks $s$ bids past this initial count,
bidders conclude that they have been outbid and become unmatched. The final
matching is communicated implicitly: the real agents observe the full published
price trajectory and simulate what good they would have been matched to had
they bid according to the published prices.

Since the private counters are noisy, more than $s$ bidders may be matched
to a good.  To maintain feasibility, the algorithm reserves some supply $m$:
i.e., it treats the supply of each good as $s-m$, rather than
$s$. The {\em reserved supply} $m$ is used to satisfy the demand of excess bidders
who believe themselves to be matched to a good; the number of such
bidders is at most $s$, with high probability.

Our algorithm stops as soon as fewer than $\rho n$ bidders place bids in a
round.  We show that this early stopping condition does not significantly harm
the welfare guarantee of the matching, while it substantially reduces the
{\em sensitivity} of the counters: no bidder ever bids more than $O(1/(\alpha\rho))$
times in total.  Crucially, this bound is independent of both the number of types of
goods $k$ and the number of bidders $n$. By stopping early, we greatly improve the accuracy of
the prices since the amount we must perturb the bid counts to protect
privacy increases with the sensitivity of the counters.

To privately implement the stopping condition, the algorithm maintains a separate counter
($\text{counter}_0$) which counts the number of unsatisfied bidders throughout
the run of the algorithm.  At the end of each round, bidders who are
unsatisfied will send the bit $1$ to this counter, while bidders who are matched will
send the bit $0$. If this counter increases by less than roughly $\rho n$ in any
round, the algorithm halts. (This is the $\mathbf{CountUnsatisfied}$
function.)

\begin{algorithm}[h!]
     \caption{$\pmatch(\alpha, \rho, \eps)$} \label{alg:matching}
     \begin{algorithmic}
       \STATE{\textbf{Input: }% Bidders' valuations on the goods
          Bidders' valuations
         $(\{v_{1j}\}_{j=1}^k, \ldots, \{v_{nj}\}_{j=1}^k)$}

       \STATE{\textbf{Initialize: for bidder $i$ and good $j$,}

         \begin{mathpar}
         T = \frac{8}{\alpha \rho},
         \and
         \epsilon' = \frac{\epsilon}{2T},
         \and
         E  = \frac{2\sqrt{2}}{\epsilon'} (\log{nT})^{5/2}
           \log\left(\frac{4k}{\gamma} \right),
         \and
         m = 2 E + 1
         \and
         \text{counter}_j = \textbf{Counter}(\epsilon', nT)
         \and
         p_j = c_j= 0,
         \\
         \mu(i) = \emptyset,
         \and
         d_i = 0,
         \and
         \text{counter}_0 = \textbf{Counter}(\epsilon', nT)
\end{mathpar}
       }

       \STATE{$\mathbf{Propose}$ $T$ times; \textbf{Output:} prices $p$ and allocation $\mu$.}
       \vspace{1ex}
       \STATE{\textbf{Propose:}}
       \FORALL{bidders $i$}
       \IF{$\mu(i) = \emptyset$}
       \STATE{Let $\mu(i) \in \argmax_j v_{ij} - p_j$, breaking ties arbitrarily}
       \IF{$v_{i \mu(i)} - p_{\mu(i)} \leq 0$}
       \STATE{Let $\mu(i) := \perp$ and $\textbf{Bid}(\mathbf{0})$.}
       \ENDIF
       \STATE{\textbf{else} Save $d_i := c_{\mu(i)}$ and
         $\textbf{Bid}(\mathbf{e_{\mu(i)}})$.}
       \ENDIF
       \STATE{\textbf{else} $\textbf{Bid}(\mathbf{\mathbf{0}})$}
       \ENDFOR
       \STATE{\textbf{CountUnsatisfied}}
       \vspace{1ex}
       \STATE{\textbf{Bid:} On input bid vector $\mathbf{b}$}
       \FORALL{goods $j$}
       \STATE{Feed $\mathbf{b}_j$ to $\text{counter}_j$.}
       \STATE{Update count $c_j := \text{counter}_j$.}
       \IF{$c_j \geq (p_j/\alpha + 1) (s - m)$}
       \STATE{Update $p_j := p_j + \alpha$.}
       \ENDIF
       \ENDFOR

       \STATE{}
       \STATE{\textbf{CountUnsatisfied:}}
       \FORALL{bidders $i$}
       \IF{$\mu(i) \neq \perp$ \text{ and } $c_{\mu(i)} - d_i \geq s - m$}
       \STATE{Feed $1$ to $\text{counter}_0$.}
       \STATE{Let $\mu(i) := \bot$.}
       \ENDIF
       \STATE{\textbf{else}  Feed $0$ to $\text{counter}_0$.}
       \ENDFOR
       \IF{$\text{counter}_0$ increases by less than  $\rho n - 2E$}
       \STATE{Halt and output $\mu$.}
       \ENDIF
     \end{algorithmic}
   \end{algorithm}
% \ar{The algorithm could use some comments/clarification. The role of $c$ and $d$
%   are a little mysterious at first blush, as is the halting condition in
%   CountUnsatisfied. Also ``$i$ becomes unmatched'' is used as a technical term, but
%   never ``$i$ becomes matched''. What does ``becomes unmatched'' mean exactly,
%   and how does it relate to updating the variable $g(i)$? }
% \jh{Added descriptive text above. Better?}
% \sw{changed from fewer than $1-\rho$ fraction to $\rho$ fraction}
% \jh{This last section is doing a lot. It's trying to explain three different
%   (but related) things: the private algorithm, the pseudocode, and the intuition
%   based on the non-private version. I've tried to make it do all three sort of
%   in parallel. Does it work? If not, we should refocus it.}
\subsection{Privacy Analysis}

In this section, we show that the allocation output by our algorithm satisfies
joint differential privacy with respect to any single bidder changing \emph{all}
of their valuations.  We will use a basic but useful lemma: to show joint
differential privacy, it is sufficient to show that the output sent to each
agent $i$ is an arbitrary function of (i) some global signal that is computed
under the standard constraint of differential privacy, and (ii) agent $i$'s
private data. We call this model the {\em billboard model}: agents can compute
their output by combining a common signal---as if posted on a public
billboard---with their own private data. In our case, the price history over
the course of the auction is the differentially private message posted on the
billboard. Combined with their personal private valuation, each agent can
compute their personal allocation.

\begin{lemma}[Billboard Lemma] \label{billboard}
  Suppose $\cM : \cD \rightarrow \cR$ is $(\epsilon, \delta)$-differentially
  private. Consider any set of functions $f_i : \cD_i \times \cR \rightarrow \cR'$,
  where $\cD_i$ is the portion of the database containing $i$'s data. The
  composition $\{ f_i (\Pi_i D, \cM(D)) \}$ is $(\epsilon, \delta)$-joint
  differentially private, where $\Pi_i : \cD \to \cD_i$ is the projection to
  $i$'s data.
\end{lemma}
\begin{proof}
  We need to show that for any agent $i$, the view of the other agents is
  $(\epsilon, \delta)$-differentially private when $i$'s private data is
  changed. Suppose databases $D, D'$ are $i$-neighbors, so $\Pi_j D = \Pi_j D'$
  for $j \neq i$.  Let $\cR_{-i}$ be a set of possible outputs to the bidders
  besides $i$.  Let $\cR^* = \{ r \in \cR \mid \{ f_j( \Pi_j D, r) \}_{-i} \in
  \cR_{-i} \}$.  Then, we need
  \begin{align*}
    \Pr[ \{ f_j (\Pi_j D, \cM(D)) \}_{-i} \in \cR_{-i} ]
    &\leq
    e^\epsilon \Pr[ \{ f_j (\Pi_j D', \cM(D')) \}_{-i} \in \cR_{-i} ]  + \delta \\
    &=
    e^\epsilon \Pr[ \{ f_j (\Pi_j D, \cM(D')) \}_{-i} \in \cR_{-i} ]  + \delta \\
    \text{so\ } \Pr [ \cM(D) \in \cR^* ] &\leq e^\epsilon \Pr [\cM(D') \in \cR^*] + \delta,
  \end{align*}
  but this is true since $\cM$ is $(\epsilon, \delta)$-differentially private.
\end{proof}

\iffull
\begin{restatable}{theorem}{counterpriv} \label{prices-privacy}
  The sequence of prices and counts of unsatisfied bidders released by
  \shortbreak $\pmatch(\alpha, \rho, \eps)$ satisfies $\epsilon$-differential
  privacy.
\end{restatable}
\begin{proof}[Sketch]
  % \jh{I like giving a rough intuition here, but now that we've moved composition
  %   to the appendix maybe even this will be unintelligible to the non-expert.
  %   Perhaps it's better just to skip the inuition/move it to the appendix with
  %   the rest of the proof.}
  % \ar{Lets keep the sketch.}
  We give a rough intuition here, and defer the full proof to
  \cref{counter-details}. Note that the prices can be computed from the noisy
  counts, so it suffices to show that the counts are private.  Since no bidder
  bids more than $T \approx 1/(\alpha\rho)$ times in total, the \emph{total}
  sensitivity of the $k$ price streams to a single bidder's valuations is only
  $O(1/(\alpha \rho))$ (independent of $k$) even though a single bidder could in
  principle bid $\Omega(1/\alpha)$ times on each of the $k$ streams.  Hence the
  analysis of these $k$ simultaneously running counters is akin to the analysis
  of answering {\em histogram queries}, multiple queries whose joint
  sensitivity is substantially smaller than the sum of their individual
  sensitivities.

  By setting the counter for each good with privacy parameter $\epsilon' =
  \epsilon/2T$, the prices are $\epsilon/2$ differentially private. By the
  same reasoning, setting the unsatisfied bidders counter with privacy parameter
  $\epsilon' = \epsilon/2T$ also makes the unsatisfied bidders count
  $\epsilon/2$ private. Thus, these outputs together satisfy
  $\epsilon$-differential privacy.

  While this intuition is roughly correct, there are some technical details.
  Namely, \citet{chan-counter} show privacy for a single counter with
  sensitivity $1$ on a non-adaptively chosen stream. Since intermediate
  outputs (i.e., prices) from our counters will affect the future streams (i.e.,
  future bids) for other counters, this is not sufficient. In fact, it is
  possible to prove privacy for multiple counters running on adaptively chosen
  streams, where the privacy parameter depends only on the joint sensitivity of
  the streams and not on the number of streams. We show this result using largely
  routine arguments; details can be found in \cref{counter-details}.
\end{proof}
\else
With this lemma, the privacy proof is largely routine. We defer the details to
the full version.
\fi

\begin{theorem} \label{matching-privacy}
 $\pmatch(\alpha, \rho, \eps)$ is $\epsilon$-joint differentially private.
\end{theorem}
\begin{proof}[Sketch]
  Note that given the sequence of prices, counts of unsatisfied bidders, and the
  private valuation of any bidder $i$, the final allocation to that bidder can
  be computed by simulating the sequence of bids made by bidder $i$,
  since the bids are determined by the price when bidder $i$ is slotted to bid and by
  whether the auction has halted or not. Bidder $i$'s final allocation is
  simply the final item that $i$ bids on.  The prices and halting condition are
  computed as a deterministic function of the noisy counts, which are
  $\epsilon$-differentially private \iffull by \cref{counter-error}\else \fi.
  So, \cref{billboard} shows that \pmatch is $\epsilon$-joint differentially
  private.
\end{proof}

\subsection{Utility Analysis}

In this section, we compare the weight of the matching produced by
\pmatch with OPT. As an intermediate step, we first show that the
resulting matching paired with the prices computed by the algorithm forms
an approximate matching equilibrium. We next show that any such matching must be
an approximately max-weight matching.

The so-called \emph{first welfare theorem} from general equilibrium theory
guarantees that an exact (i.e., a $(0,0,0)$-) matching equilibrium gives an exact
maximum weight matching. Compared to this ideal, \pmatch loses
welfare in three ways. First, a $\rho$ fraction of bidders may end up
unsatisfied.  Second, the matched bidders are not necessarily matched to goods
that maximize their utility given the prices, but only to goods that do so
approximately (up to additive $\alpha$).  Finally, the auction sets aside part
of the supply to handle over-allocation from the noisy counters. This reserved
supply may end up unused, say, if the counters are accurate or actually
under-allocate. In other words, we compute an equilibrium of a market with reduced
supply, so our welfare guarantee holds if the supply $s$ is significantly
larger than the necessary reserved supply $m$.

The key performance metric is \emph{how much} supply is needed to achieve a
given welfare approximation in the final matching. On the one hand, we will show
later that the problem is impossible to solve privately if $s = O(1)$
(\cref{sec:lowerbounds}).  On the other hand, the problem is trivial if $s \geq
n$: agents can be simultaneously matched to their favorite good with no
coordination; this allocation is trivially both optimal and private. Our
algorithm will achieve positive results in the intermediate supply range, when
$s \geq \polylog(n)$.
% \sw{Do you mean $\geq$ here?}
% \jh{Yes. Fixed.}

\begin{theorem} \label{matching-welfare}
  Let $\alpha>0$, and $\mu$ be the matching computed by
  $\pmatch(\alpha/3, \alpha/3, \eps)$. Let $\OPT$ denote the weight of the
  optimal matching. Then, if the supply satisfies
  \[
    s \geq \frac{16 E' + 4}{\alpha} = O\left( \frac{1}{\alpha^3 \epsilon}
      \cdot \polylog \left( n, k, \frac{1}{\alpha}, \frac{1}{\gamma}
      \right)\right),
  \]
  and $n > s$,  the matching $\mu$ has social welfare at least $\OPT - \alpha
  n$ with probability $\geq 1-\gamma$, where
   % \begin{mathdisplayfull}
   \[
     E' = \frac{288\sqrt{2}}{\alpha^2 \epsilon}
     \left(\log\left(\frac{72n}{\alpha^2} \right)\right)^{5/2}
     \log\left(\frac{4k}{\gamma} \right).
   \]
   % \end{mathdisplayfull}%
\end{theorem}
% \ar{Is $\beta$ a parameter that has been given a numerical value? I'm not sure
%   how to interpret $s \geq 6\beta/\alpha$. Do we mean $E$ here instead of
%   $\beta$?}
% \sw{hmm, the denominator should be $\alpha/3$ now}
% \jh{Unsure. Is this fixed now or still open?}
% \sw{should be fixed now, including the constants}
\begin{remark}
  Our approximation guarantee here is \emph{additive}. Later in this section, we
  show that if we are in the \emph{unweighted} case---$v_{ij} \in \{0,1\}$---we
  can find a matching $\mu$ with welfare at least $(1-\alpha)\mathrm{OPT}$.
  This \emph{multiplicative} guarantee is unusual for a differentially private
  algorithm.
\end{remark}

The proof follows from the following lemmas.
\iffull\else (We defer some proofs to the full version.) \fi

\begin{lemma} \label{matching-eq-alpha}
  We call a bidder who wants to continue bidding {\em unsatisfied}; otherwise
  bidder $i$ is {\em satisfied}. At termination of \shortbreak $\pmatch(\alpha,\rho,\eps)$,
  all satisfied bidders $i$ are matched to a good $\mu(i)$ such that
  \[
    v_{i,\mu(i)} - p_{\mu(i)} \geq \max_j (v_{i,j} - p_j) - \alpha .
  \]
\end{lemma}
\iffull
\begin{proof}
  Fix any satisfied bidder $i$ matched to $j^* = \mu(i)$. At the time that
  bidder $i$ last bid on $j^*$, by construction, $v_{ij^*} - p_{j^*} \geq
  \max_{j}(v_{ij}-p_j)$.  Since $i$ remained matched to $j^*$, its price could
  only have increased by at most $\alpha$, and the prices of other goods $j \neq
  j^*$ could only have increased.  Hence, at completion of the algorithm,
  \begin{mathdisplayfull}
    v_{i,\mu(i)} - p_{\mu(i)} \geq \max_{j}(v_{ij}-p_j) - \alpha
  \end{mathdisplayfull}%
  for all matched bidders $i$.
\end{proof}
\fi

\begin{lemma} \label{matching-eq-beta}
  Assume all counters have error at most $E$ throughout the run of
  $\pmatch(\alpha, \rho, \eps)$. Then the number of bidders assigned to any good
  is at most $s$ and each over-demanded good clears except for at most $\beta$
  supply, where
  \[
    \beta = 4 E + 1=  O\left(\frac{1}{\alpha\rho \epsilon}\cdot \polylog \left(
        \frac{1}{\alpha}, \frac{1}{\rho}, \frac{1}{\gamma},k,n \right)\right).
  \]
\end{lemma}
\begin{proof}
Since the counter for each under-demanded good never exceeds $s-m$, we know that
each under-demanded good is matched to no more than $s-m+E < s$ bidders.
Consider any counter $c$ for an over-demanded good.  Let $t$ be a time step such
that
\begin{mathdisplayfull}
 c(nT) - c(t + 1) \leq s - m < c(nT) - c(t) ,
\end{mathdisplayfull}%
where $c(t)$ denotes the output of the counter at time $t$.
Note that the bidders who bid after time $t$ are the only bidders matched to
this good  at time $nT$. Let $\sigma$ be the true bid stream for this good and
let the sum of bids in $\sigma$ up to time $t$ be $h(\sigma, t)$. Then,
the total number of bidders allocated to this good at time $nT$ is
\begin{align*}
 h(\sigma, nT) - h(\sigma, t) & \leq h(\sigma, nT) - h(\sigma, t + 1) + 1 \\
 & \leq (c(nT) + E) -  (c(t + 1) - E) + 1 \\
 & \leq s - m + 2E + 1 =  s.
\end{align*}
  Similarly, we can lower bound the number of bidders allocated to this good:
  \begin{align*}
    h(\sigma, nT) - h(\sigma, t)
    & = (h(\sigma,nT) - c(nT)) + (c(nT) - c(t)) + (c(t) - h(\sigma, t)) \\
    & > s - m - 2E > s - 4E - 1.
  \end{align*}
  Therefore, every over-demanded good clears except for at most $\beta =
  4E + 1$ supply, which gives
  \begin{align*}
    \beta &=  \frac{16\sqrt{2}}{\alpha \rho \epsilon} \left(\log\left(\frac{6n}{\alpha\rho} \right)\right)^{5/2}
    \log\left(\frac{4k}{\gamma} \right) + 1 \\
    &=
    O\left(\frac{1}{\alpha\rho \epsilon}\cdot \polylog \left(
        \frac{1}{\alpha}, \frac{1}{\rho}, \frac{1}{\gamma},k,n \right)\right).
  \end{align*}
  \ar{Let's put the actual dependence on all terms in, rather than writing polylog.
    This should also go in the theorem statement.}
\end{proof}

\begin{lemma} \label{matching-eq-rho}
  Assume all counters have error at most $E$ throughout the run of
  $\pmatch(\alpha, \rho, \eps)$. Then at termination all but a $\rho$ fraction of
  bidders are satisfied, so long as $s \geq 8 E + 1$ and $n \geq 8E/\rho$.
\end{lemma}
\begin{proof}
  First, we show that the total number of bids made over the course of the
  algorithm is bounded by $3n/\alpha$.
  %\ar{Do we ever define what a ``bid'' is? We should add some annotation to the
  % algorithm to specify when a ``bid'' is said to have been made.}
  We account separately for the under-demanded goods (those with price 0 at the
  end of the auction) and the over-demanded goods (those with positive price).
  For the under-demanded goods, since their prices remain 0 throughout the
  algorithm, their corresponding noisy counters never exceeded $(s-m)$.
  Since no bidder is ever unmatched after having been matched to an
  under-demanded good, the set of under-demanded goods can receive at most one bid
  from each agent; together the under-demanded goods can receive at most $n$
  bids.
  % \ar{Explain why. Also, similar to ``bid'', we should annotate the algorithm to
  %   specify when a bidder is ``rejected''.}
  % \sw{ Bid is now written as a function}

  Next, we  account for the over-demanded goods. Note that the bidders matched to
  these goods are precisely the bidders who bid within $s- m$ ticks of the final
  counter reading. Since the counter has error bounded by $E$ at each time step,
  this means at least $s - m - 2E$ bidders end up matched to each over-demanded
  good.  Since no agent can be matched to more than one good there can be at
  most $n/(s-m-2E)$ over-demanded goods in total.

  Likewise, we can account for the number of price increases per over-demanded
  good. Prices never rise above $1$ (because any bidder would prefer to be
  unmatched than to be matched to a good with price higher than $1$).  Therefore,
  since prices are raised in increments of $\alpha$, the price of every
  over-demanded good increases at most $1/\alpha$ times.  Since there can be at
  most $(s - m + 2E)$ bids between each price update (again, corresponding to $s
  - m$ ticks of the counter), the total number of bids received by all of the
  over-demanded goods in total is at most
  \[
    \frac{n}{s-m-2E}\cdot \frac{1}{\alpha}\cdot (s-m+2E).
  \]
  Since each bid is either on an under or over-demanded good, we can upper
  bound the \emph{total} number of bids $B$ by
  \[
    B \leq n + \frac{n}{\alpha} \left( \frac{s - m + 2E}{s - m - 2E}\right) =
    \frac{n}{\alpha} \left(\alpha + \frac{s- m +2E}{ s-m-2E}\right).
  \]
  The algorithm sets the reserved supply to be $m = 2 E + 1$ and by assumption,
  we have $s \geq 8E+1$. Since we are only interested in cases where $\alpha <
  1$, we conclude
  \begin{equation} \label{match-bid-ub}
    B \leq n +   \frac{n}{\alpha} \left( \frac{s - m + \alpha}{s - m -
        \alpha} \right) \leq \frac{3n}{\alpha}.
  \end{equation}

  Now, consider the halting condition. Either the algorithm
  halts early, or it does not. We claim that at termination, at most $\rho n $
  bidders are unsatisfied. The algorithm halts early if at any round of
  \textbf{CountUnsatisfied}, $\text{counter}_0$ (which counts the number of unsatisfied
  bidders) increases by less than $\rho n - 2E$, when there are at most $\rho n
  - 2 E + 2E = \rho n$ unsatisfied bidders.

  Otherwise, suppose the algorithm does not halt early.  At the start of each
  round there must be at least $\rho n - 4E$ unsatisfied bidders. Not all of
  these bidders must bid during the \textbf{Propose} round since price increases
  while they are waiting to bid might cause them to no longer demand any item,
  but this only happens if bidders prefer to be unmatched at the new prices.
  Since prices only increase, these bidders remain satisfied for the rest of the
  algorithm.  If the algorithm runs for $R$ rounds and there are $B$ true bids,
  \begin{mathdisplayfull}
    B \geq R (\rho n - 4E) - n.
  \end{mathdisplayfull}%
  Combined with our upper bound on the number of bids (\cref{match-bid-ub}) and
  our assumption $\rho n \geq 8E$, we can upper bound the number of rounds $R$:
  \[
    R \leq \left(\frac{3n}{\alpha} + n\right) \cdot \left( \frac{1}{\rho n - 2E}
    \right) \leq \left(\frac{4n}{\alpha}\right) \left(\frac{2}{\rho n}\right) =
    \frac{8}{\alpha \rho} := T .
  \]
  Thus, running the algorithm for $T$ rounds leads to all but $\rho n$ bidders
  satisfied.
\end{proof}

\begin{lemma} \label{matching-acc}
  With probability at least $1 - \gamma$, $\pmatch(\alpha, \rho, \eps)$ computes an
  $(\alpha, \beta, \rho)$-matching equilibrium, where
  \[
    \beta = 4E+1 = O\left(\frac{1}{\alpha\rho \epsilon}\cdot \polylog \left(
        \frac{1}{\alpha}, \frac{1}{\rho}, \frac{1}{\gamma},k,n \right)\right)
  \]
  so long as $s \geq 8E + 1 \mbox{ and } n \geq 8E/\rho$.
\end{lemma}
\iffull
\begin{proof}
  By \cref{counter-error}, $\text{counter}_0$ is $\left( \lambda_1,
    \gamma/2\right)$-useful, and each of the $k$ good counters is
  $\left(\lambda_2 , \gamma/2  \right)$-useful, where
\[
  \lambda_1 = \frac{2\sqrt{2}}{\epsilon'}
  (\log{nT})^{5/2}\log\left(\frac{4}{\gamma} \right)
  \quad \text{and} \quad
  \lambda_2 = \frac{2\sqrt{2}}{\epsilon'} (\log{nT})^{5/2}
  \log\left(\frac{4k}{\gamma} \right).
\]
 Since we set $E = \lambda_2 > \lambda_1$, all counters are $(E,
 \gamma/2)$-useful, and thus with probability at least $1 - \gamma$, all
 counters have error at most $E$. The theorem then follows by
 \cref{matching-eq-alpha,matching-eq-beta,matching-eq-rho}.

 \end{proof}
\fi

With these lemmas in place, it is straightforward to prove the welfare theorem
(\cref{matching-welfare}).
\begin{proof}[\cref{matching-welfare}]
  By \cref{matching-acc}, $\pmatch(\alpha/3, \alpha/3, \epsilon)$
  calculates a matching $\mu$ that is an $(\alpha/3, \beta,
  \alpha/3)$-approximate matching equilibrium with probability at least
  $1-\gamma$, where $\beta = 4E' + 1$.  Let $p$ be the prices at the end of the
  algorithm, and $S$ be the set of satisfied bidders.  Let $\mu^*$  be the
  optimal matching achieving welfare $\sum_{i=1}^n v_{i,\mu^*(i)} =
  \mathrm{OPT}$.  We know that $|S|\geq (1-\alpha/3)n$ and
  \[
    \sum_{i\in S} (v_{i \mu(i)} - p_{\mu(i)}) \geq \sum_{i\in
    S}(v_{i\mu^*(i)} - p_{\mu^*(i)}) - \alpha|S|/3.
  \]
  Let $N^*_j$ and $N_j$ be the number of goods of type $j$ matched in $\mu^*$
  and $\mu$ respectively, and let $G$ be the set of over-demanded goods at prices $p$.

  Since each over-demanded good clears except for at most $\beta$ supply, and
  since each of the $n$ agents can be matched to at most one good, we know that
  $|G|\leq n/(s-\beta)$. Since the true supply in $\OPT$ is at most $s$, we also
  know that $N^*_j - N_j \leq \beta$ for each over-demanded good $j$. Finally, by
  definition, under-demanded goods $j$ have price $p_j = 0$. So,
  \begin{align*}
    \sum_{i \in S} v_{i\mu^*(i)} - \sum_{i \in S} v_{i\mu(i)}
    & \leq  \sum_{i\in S} p_{\mu^*(i)} - \sum_{i\in S} p_{\mu(i)} + \alpha|S|/3 \\
    & = \sum_{j \in G} p_j (N^*_j - N_j) + \alpha |S|/3 \\
    & \leq \sum_{j \in G}  \beta  + \alpha |S|/3 \leq \frac{n \beta}{s-\beta} +
    \alpha |S|/3.
  \end{align*}
  If $s \geq 4\beta/\alpha$, the first term is at most $\alpha n/3$. Finally,
  since all but $\alpha n/3$ of the bidders are matched with goods in $S$, and
  their valuations are upper bounded by $1$, so
  \[
    \sum_i v_{i\mu(i)} - \sum_{i} v_{i\mu^*(i)} \leq \alpha n/3  + \alpha|S|/3 +
    \alpha n/3 \leq \alpha n .
  \]
  Unpacking $\beta$ from \cref{matching-acc}, we get the stated bound on supply.
\end{proof}

\subsection{Multiplicative Approximation to Welfare}
In certain situations, a slight variant of \pmatch (\cref{alg:matching}) can
give a multiplicative welfare guarantee. In this section, we will assume that
the value of the maximum weight matching $\OPT$ is known; it is often possible
to privately estimate this quantity to high accuracy. Our algorithm is \pmatch
with a different halting condition: rather than count the number of unmatched
bidders each round, count the number of bids per round. Once this count drops
below a certain threshold, halt the algorithm.

More precisely, we use a function $\mathbf{CountBids}$ (\cref{alg:count-bids})
in place of \longbreak $\mathbf{CountUnsatisfied}$ in \cref{alg:matching}.

\begin{algorithm}[h!]
     \caption{Modified Halting Condition $\mathbf{CountBids}$}
     \begin{algorithmic}\label{alg:count-bids}
       \STATE{\textbf{CountBids:}}
       \FORALL{bidders $i$}
       \IF{$\mu(i) \neq \perp$ \text{ and } $c_{\mu(i)} - d_i \geq s - m$}
       \STATE{Let $\mu(i) := \emptyset$}
       \ENDIF
       \IF{$i$ bid this round}
       \STATE{Feed $1$ to $\text{counter}_0$.}
       \ENDIF
       \STATE{\textbf{else}  Feed $0$ to $\text{counter}_0$.}
       \ENDFOR
       \IF{$\text{counter}_0$ increases by less than  $\frac{\alpha
           OPT}{2\lambda} - 2E$}
       \STATE{Halt; For each $i$ with $\mu(i) = \emptyset$, let $\mu(i) =
         \perp$}
       \ENDIF
     \end{algorithmic}
   \end{algorithm}

\begin{theorem} \label{mult-acc}
  Suppose bidders have valuations $\{ v_{ij} \}$ over goods such that
  \begin{mathdisplayfull}
    \min_{ v_{ij} > 0 } v_{ij} \geq \lambda.
  \end{mathdisplayfull}%
  Then \cref{alg:matching}, with
  \begin{mathdisplayfull}
    T = \fullfrac{24}{\alpha^2}
  \end{mathdisplayfull}%
  rounds, using stopping condition $\mathbf{CountBids}$ (\cref{alg:count-bids})
  in place of \longbreak $\mathbf{CountUnsatisfied}$ and stopped once the total bid
  counter increases by less than
  \begin{mathdisplayfull}
    \fullfrac{\alpha \OPT}{2 \lambda} - 2E
  \end{mathdisplayfull}
  bids in a round, satisfies $\epsilon$-joint differential  privacy
  and outputs a matching that has welfare at least $O((1 -
  \alpha/\lambda)\OPT)$, so long as
  \[
    s = \Omega \left( \frac{1}{\alpha^3 \epsilon} \cdot \polylog \left( n, k,
        \frac{1}{\alpha}, \frac{1}{\gamma} \right) \right)
\iffull\else
  \]
  \[
\fi
    \text{and} \qquad
    \OPT = \Omega \left( \frac{\lambda}{ \alpha^3 \epsilon} \cdot \polylog
      \left( n, k, \frac{1}{\alpha}, \frac{1}{\gamma} \right) \right).
  \]
  % \jh{NOTE: need bidders to not bid on valuation $0$ stuff. Maybe we should just
  %   make that a requirement in the main algorithm? Only bid if strictly better
  %   than not bidding.}
\end{theorem}
\iffull
\begin{proof}
  Privacy follows exactly like \cref{matching-privacy}. We first show that at
  termination, all but $\alpha \OPT /\lambda$ bidders are matched to an
  $\alpha$-approximate favorite item. The analysis is very similar to
  \cref{matching-acc}.  Note that every matched bidder is matched to an
  $\alpha$-approximate favorite good, since it was an exactly favorite good at
  the time of matching, and the price increases by at most $\alpha$. Thus, it
  remains to bound the number of unsatisfied bidders at termination.

  Condition on all counters having error bounded by $E$ at all time steps; by
  \cref{counter-error} and a union bound over counters, this happens with
  probability at least $1 - \gamma$. Like above, we write $s' = s - m$ for the
  effective supply of each good. Let us first consider the case where the
  algorithm stops early. If the total bid counter changes by less than
  $\frac{\alpha \OPT}{2\lambda} - 2E$, the true number of bids that round is at
  most
  \[
    Q = \frac{\alpha \OPT}{2\lambda}.
  \]

  We will upper bound the number of unsatisfied bidders at the end of the round.
  Note that the number of unsatisfied bidders at the end of the round is the
  number of bidders who have been rejected in the current round. Suppose there
  are $N$ goods that reject bidders during this round. The total count on these
  goods must be at least
  \[
    (s' - 2E) \cdot N - Q
  \]
  at the start of the round, since each counter will increase by at most $2E$
  due to error, and there were at most $Q$ bids this round. By our conditioning,
  there were at least
  \[
    (s' - 2E) \cdot N - Q - 2EN
  \]
  bidders matched at the beginning of the round. Since bidders are only matched
  when their valuation is at least $\lambda$, and the optimal weight matching is
  $\OPT$, at most $\frac{OPT}{\lambda}$ bidders can be matched at any time.
  Hence,
  \[
    N \leq \left( \frac{\OPT}{\lambda} + Q \right) \cdot \frac{1}{s' - 4E}.
  \]
  Then, the total number of bidders rejected this round is at most $2EN + Q$.
  Simplifying,
  \begin{align*}
    2EN + Q & \leq \frac{2E}{s' - 4E} \cdot \left( \frac{\OPT}{\lambda} + Q \right) + Q \\
            & \leq \left( \frac{6E}{s' - 4E}\right) \left(\frac{\OPT}{\lambda} \right)
            + \frac{\alpha \OPT}{2\lambda}.
  \end{align*}
  To make the first term at most $\frac{\alpha \OPT}{2\lambda}$, it suffices to
  take
  \begin{align*}
    \frac{6E}{s' - 4E} & \leq \frac{\alpha}{2} \\
    s'                 & \geq \frac{12 E}{\alpha} + 4E \\
    s                  & \geq \frac{12 E}{\alpha} + 6E + 1,
  \end{align*}
  or $s \geq 18E/\alpha$. In this case, the algorithm terminates with at most
  $\frac{\alpha \OPT}{\lambda}$ unsatisfied bidders, as desired.

  On the other hand, suppose the algorithm does not terminate early, the bid
  count increasing by at least $Q - 2E$ every round. By our conditioning, this
  means there are at least $Q - 4E$ bids each round; let us bound the number of
  possible bids.

  Since bidders only bid if they have valuation greater than $\lambda$ for a
  good, and since the maximum weight matching has total valuation $\OPT$, at
  most $\OPT/\lambda$ bidders can be matched. Like before, we say goods are
  under-demanded or over-demanded: they either have final price $0$, or positive
  final price.

  There are at most $\OPT/\lambda$ true bids on the goods of the first type;
  this is because bidders are never rejected from these goods. Like before,
  write $s' = s - m$. Each counter of a over-demanded good shows $s'$ people
  matched, so at least $s' - 2E$ bidders end up matched.  Thus, there are at
  most
  \begin{mathdisplayfull}
    \fullfrac{\OPT}{\lambda (s' - 2E)}
  \end{mathdisplayfull}%
  over-demanded goods. Each such good takes at most $s' + 2E$ bids at each of
  $1/\alpha$ price levels. Putting these two estimates together, the total
  number of bids $B$ is upper bounded by
  \[
    B \leq \frac{\OPT}{\lambda} \cdot \left( 1+ \frac{s' + 2E}{s' - 2E} \right)
    \leq \frac{6 \OPT}{\lambda \alpha}
  \]
  if $s' \geq 4E$, which holds since we are already assuming $s' \geq 4E +
  \frac{12E}{\alpha}$. Hence, we know the number of bids is at most
  \begin{align*}
    T \cdot (Q - 4E) &\leq B \leq \frac{6 \OPT}{\lambda \alpha} \\
    T &\leq \frac{6 \OPT}{\lambda} \cdot \left( \frac{2 \lambda}{\alpha \OPT - 8
        \lambda E} \right).
  \end{align*}
  Assuming $\alpha \OPT \geq 16 \lambda E$, we find $T \leq 24/\alpha^2$.

  With this choice of $T$, the supply requirement is
  \begin{equation} \label{eq-mult-supply}
    s \geq \frac{18E}{\alpha} = \Omega \left( \frac{1}{\alpha^3 \epsilon} \cdot
      \polylog \left( n, k, \frac{1}{\alpha}, \frac{1}{\gamma} \right) \right).
  \end{equation}
  Likewise, the requirement on $\OPT$ is
  \[
    \OPT \geq \frac{16\lambda E}{\alpha} = \Omega \left( \frac{\lambda}{
        \alpha^3 \epsilon} \cdot \polylog \left( n, k, \frac{1}{\alpha},
        \frac{1}{\gamma} \right) \right).
  \]

  Now, we can follow the analysis from \cref{matching-welfare} to bound the
  welfare. Suppose the algorithm produces a matching $\mu$, and consider any
  other matching $\mu^*$.  For each bidder who is matched to an
  $\alpha$-approximate favorite good,
  \begin{mathdisplayfull}
    v_{i \mu(i)} - p_{\mu(i)} \geq v_{i \mu^*(i)} - p_{\mu^*(i)} - \alpha.
  \end{mathdisplayfull}%
  Each such bidder is matched to a good with value at least $\lambda$, so there
  are at most $\OPT/\lambda$ such bidders. Summing over these bidders (call them
  $S$),
  \[
    \sum_{i \in S} v_{i \mu(i)} - p_{\mu(i)} \geq \sum_{i \in S} v_{i \mu^*(i)}
    - p_{\mu^*(i)} - \frac{\alpha \OPT}{\lambda}.
  \]
  Letting $N_j, N_j^*$ be the number of goods of type $j$ matched in $\mu,
  \mu^*$ and rearranging,
  \[
    \sum_{i \in S} v_{i \mu^*(i)} - v_{i\mu(i)} \leq \sum_{j \in S} p_j(N_j^* -
    N_j) + \frac{\alpha \OPT}{\lambda}.
  \]
  Exactly the same as in \cref{matching-welfare}, each over-demanded good $(p_j >
  0)$ clears except for at most $\beta = 4E + 1$ supply.  Since at most
  $\frac{\OPT}{\lambda}$ bidders can be matched, the number of goods with $p_j >
  0$ is at most
  \[
    \frac{\OPT}{\lambda (s - \beta)}.
  \]
  Like before, $N_j^* - N_j \leq \beta$. Since there are at most $\alpha \OPT /
  \lambda$ bidders not in $S$ and each has valuation in $[0, 1]$, when summing
  over all bidders,
  \[
    \sum_{i} v_{i \mu^*(i)} - v_{i\mu(i)} \leq \frac{\OPT \beta}{\lambda (s -
      \beta)} + \frac{\alpha \OPT}{\lambda} + \frac{\alpha \OPT}{\lambda}.
  \]
  The first term is at most $\alpha \OPT / \lambda$ for $s \geq \beta (1 +
  1/\alpha)$, when the algorithm calculates a matching with weight $O( (1 -
  \alpha/\lambda) \OPT)$. Since $\beta = 4E + 1$, this reduces to the supply
  constraint \cref{eq-mult-supply}.
\end{proof}
\else
  Privacy follows like \cref{matching-privacy}. Utility follows a similar
  analysis as for the matching case, with one main twist: in the unwweighted
  case, there can be at most $\OPT/\lambda$ bidders matched to a prefered good,
  since each matched bidder contributes weight $\lambda$. Thus, we can halt the
  algorithm sooner when $\OPT$ is small. Details can be found in the full
  version.
\fi
\begin{remark}
  For a comparison with \cref{matching-welfare} and \pmatch, consider the
  ``unweighted'' case where bidders have valuations in $\{0, 1\}$ (i.e.,
  $\lambda = 1$). Note that both \pmatch and the multiplicative version require
  the same lower bound on supply. Ignoring log factors, \pmatch requires $n =
  \tilde{\Omega}(1/\alpha^3 \epsilon)$ for an additive $\alpha n$ approximation,
  while \cref{mult-acc} shows $\OPT = \tilde{\Omega}(1/\alpha^3 \epsilon)$ is
  necessary for a multiplicative $\alpha$, hence additive $\alpha \OPT$,
  approximation. Hence, \cref{mult-acc} gives a stronger guarantee if $\OPT =
  \tilde{o}(n)$ in the unweighted case, ignoring log factors.
\end{remark}

\section{Extension to Gross Substitute Valuations}
\label{sec:extensions}
% \jh{Danger: note that $\OPT$ can't be too small for multiplicative error.}
% \sw{Should define $d$ in this paragraph.}
% \jh{Defined below. I prefer not to talk about $d$ in the privacy proof, it's
%   unnecessary information.}

% \subsection{Gross Substitute Valuations}
% \sw{did some modifications to get $\OPT - \alpha d$; saved the
%   old setting to another file.}
While Kelso and Crawford's algorithm is simplest in the unit demand setting, it
can also compute allocations when bidders have gross substitutes valuations.
Before we discuss our analogous extension, we will first introduce some notation
for gross substitutes valuations.  Unlike unit demand valuations, bidders with
gross substitute valuations may demand more than one good. Let $\Omega = 2^G$
denote the space of bundles (i.e., subsets of goods). Like previous sections,
let $k$ be number of types of goods, and let $s$ be the supply of each type of
good.  Let $d$ denote the {\em market size}---the total number of goods,
including identical goods, so $d = ks$.\footnote{%
  In general, goods may have different supplies, if $s$ denotes the
  \emph{minimum} supply of any good. Hence, $d$ is not necessarily dependent on $s$.}
% \sw{We might sort of treat $d$ as an independent parameter (in the more general
%   model, supplies might be different for different goods)}
We assume that each bidder has a valuation function on bundles, $v_i : \Omega
\rightarrow [0,1]$, and that this valuation satisfies the gross substitutes
condition (\cref{def-gs}).

Like before, we simulate $k$ ascending price auctions in rounds.  Bidders now
maintain a bundle of goods that they are currently allocated to, and bid on one
new good each round. For each good in a bidder's bundle, the bidder keeps track
of the count of bids on that good when it was added to the bundle. When the current
count ticks past the supply, the bidder knows that they have been outbid.

The main subtlety is in how bidders decide which goods to bid on. Namely, each
bidder treat goods in their bundle as fixed in price (i.e., bidders ignore
the price increment of at most $\alpha$ that might have occurred after winning
the item). Goods outside of their bundle (even if identical to goods in their
bundle) are evaluated at the true price. We call these prices the bidder's {\em
  effective} prices, so each bidder bids on an arbitrary good in his
most-preferred bundle at the effective prices. The full algorithm is given in
\cref{gs-auction}.

\begin{algorithm}[ht!]
     \caption{$\palloc(\alpha, \rho, \eps)$ (with Gross Substitute Valuations)}
     \begin{algorithmic}\label{gs-auction}

       \STATE{\textbf{Input:} Bidders' gross substitute valuations on the
       bundles $\{ v_i : \Omega \rightarrow [0, 1] \}$}

       \STATE{\textbf{Initialize: for bidder $i$ and good $j$,}

         \begin{mathpar}
         T = \frac{10}{\alpha \rho},
         \and
         \epsilon' = \frac{\epsilon}{2T},
         \and
         E = \frac{2\sqrt{2}}{\epsilon'} (\log nT)^{5/2} \log \left(
           \frac{4k}{\gamma} \right)  + 1,
        \and
        m = 2E + 1,
        \and
        \text{counter}_0 = \counter(\epsilon', nT),
        \and
         \text{counter}_j = \counter(\epsilon', nT),
         \and p_j = c_j = 0,
         \and
         d_g = 0,
         \and
         g(i) = \{ \emptyset \} \and \text{for every bidder\ } i
\end{mathpar}
}

       \STATE{$\mathbf{Propose}$ $T$ times; \textbf{Output:} prices $p$ and allocation $g$.}
       \vspace{1ex}
       \STATE{\textbf{Propose:}}
       \FORALL{bidders $i$}
       \FORALL{goods $g \in g(i)$}
       \IF{$c_{type(g)} - d_g \geq s - m$}
       \STATE{Remove $g(i) := g(i) \setminus g$}
       \ENDIF
       \ENDFOR
       \STATE{Let $p_0$ be the original cost of $g(i)$.}
       \STATE{Let $\omega^* \in \displaystyle\argmax_{\omega \supsetneq g(i)} {v_{i}(\omega) -
           p(\omega \setminus g(i)) - p_0}$ arbitrary.}
       \IF{$v_{i}(\omega^*) - p(\omega \setminus g(i)) -  p_0 \geq v_i(g(i)) -
         p_0$}
       \STATE{Let $j \in \omega^* \setminus g(i)$ arbitrary.}
       \STATE{Save $d_j := c_{type(j)}$}
       \STATE{Add $g(i) := g(i) \cup j$ and $\textbf{Bid}(\mathbf{e_j})$}
       \ENDIF
       \STATE{\textbf{else} $\textbf{Bid}(\mathbf{0})$}
       \ENDFOR
       \STATE{\textbf{CountUnsatisfied}}
       \vspace{1ex}
       \STATE{\textbf{Bid:} On input bid vector $\mathbf{b}$}
       \FORALL{goods $j$}
       \STATE{Feed $\mathbf{b}_j$ to $\text{counter}_j$.}
       \STATE{Update count $c_j := \text{counter}_j$.}
       \IF{$c_j \geq (p_j/\alpha + 1) (s - m)$}
       \STATE{Update $p_j := p_j + \alpha$.}
       \ENDIF
       \ENDFOR
       \STATE{}
       \STATE{\textbf{CountUnsatisfied:}}
       \FORALL{bidders $i$}
       \IF{ $i$ wants continue bidding}
       \STATE{Feed $1$ to $\text{counter}_0$.}
       \ENDIF
       \STATE{\textbf{else} Feed $0$ to $\text{counter}_0$.}
       \ENDFOR
       \IF{$\text{counter}_0$ increases by less than $\rho d - 2E$}
       \STATE{Halt and output $\mu$.}
       \ENDIF
     \end{algorithmic}
   \end{algorithm}

   Privacy is very similar to the case for matchings.

\begin{theorem} \label{gs-priv}
$\palloc(\alpha, \rho, \eps)$ satisfies $\eps$-joint differential privacy.
\end{theorem}
\iffull
\begin{proof}
  Essentially the same proof as \cref{matching-privacy}.
\end{proof}
\fi

\begin{theorem} \label{gs-welfare}
  Let $0<\alpha< n/d$, and  $g$ be the allocation computed by $\palloc(\alpha/3,
  \alpha/3, \eps)$, and let $\OPT$ be the optimum max welfare.  Then, if $d \geq
  n$ and
  \[
    s \geq \frac{12E' + 3}{\alpha} = O \left( \frac{1}{\alpha^3 \epsilon}
    \cdot \polylog\left( n, k, \frac{1}{\alpha}, \frac{1}{\gamma} \right) \right),
  \]
  the allocation $g$ has social welfare at least
  \[
  \sum_{i=1}^n v_i(g(i)) \geq \OPT - \alpha d
  \]
  with probability at  least $1 - \gamma$, where
  \[
    E' = \frac{360 \sqrt{2}}{\alpha^2
      \eps}\left(\log\left(\frac{90n}{\alpha^2} \right)
    \right)^{5/2}\log\left(\frac{4k}{\gamma}\right) + 1.
  \]
\end{theorem}

\begin{remark}
In comparison with \cref{matching-welfare}, \cref{gs-welfare} requires a similar
constraint on supply but promises welfare $\OPT - \alpha d$ rather than
$\OPT - \alpha n$. Since $\OPT \leq n$ this guarantee is only non-trivial for
$\alpha \leq n/d$, so the supply has a polynomial dependence on the total
size of the market $d$. In contrast, \cref{matching-welfare} guarantees good
welfare when the supply has a \emph{logarithmic} dependence on the total number of
goods in the market.

We note that if bidders demand bundles of size at most $b$, then we can
improve the above welfare bound to $\OPT - \alpha n b$. Note that this is
independent of the market size $d$ and smoothly generalizes the matching case
where $b = 1$. 
\end{remark}

Similar to \cref{matching-eq}, we define an \emph{approximate allocation
  equilibrium} as a prerequisite for showing our welfare guarantee.

\begin{definition} \label{alloc-eq}
  A price vector $p\in [0,1]^k$ and an assignment $g\colon [n] \rightarrow
  \Omega$ of bidders to goods is an {\em $(\alpha, \beta, \rho)$-approximate
  allocation equilibrium} if
  \begin{enumerate}
    \item for all but $\rho d$  bidders, $v_i(g(i)) - p(g(i)) \geq
      \max_{\omega \in \Omega} v_i(\omega) - p(\omega) - \alpha |g(i)|$;
    \item the number of bidders assigned to any good is at most $s$; and
    \item each over-demanded good clears except for at most $\beta$ supply.
  \end{enumerate}
\end{definition}
% \ar{What is $\Omega$ and where is it defined? Shouldn't $g$ map bidders to
%   subsets of goods?}
% \jh{Was in the prior definition, now removed.}
%
The following lemmas show that our algorithm finds an approximate allocation
equilibrium.
\iffull
We prove the last two requirements first.
\else
(We defer proofs to the full version.)
\fi
\begin{lemma} \label{gs-supply}
  Assume all counters have error at most $E$ throughout the run of
  $\palloc(\alpha,\rho,\eps)$. Then, the number of bidders assigned to any good is at most
  $s$ and each over-demanded good clears except for at most $\beta$ supply,
  where
  \[
    \beta = 4E + 1 = O \left( \frac{1}{\alpha \rho \epsilon } \cdot \polylog
      \left( n, k, \frac{1}{\alpha}, \frac{1}{\rho}, \frac{1}{\gamma} \right)
    \right).
  \]
\end{lemma}
\iffull
\begin{proof}
  Consider any good $j$. If it is under-demanded, the counter corresponding to
  $j$ never rise above $s - m$. Hence by our conditioning, at most $s - m + E <
  s$ bidders are assigned to $j$.  If $j$ is over-demanded, the same reasoning as
  in \cref{matching-acc} shows that the number of bidders matched to $j$ lies in
  the range $[s - m - 2E, s - m + 2E + 1]$. By the choice of $m$, the upper
  bound is at most $s$. Likewise, at least $s - m + E = s - (4E + 1)$ bidders
  are assigned to $j$.  Setting $\beta =  4E + 1$ gives the desired bound.
\end{proof}
\fi

\begin{lemma} \label{gs-alpha}
  We call a bidder who wants to bid more {\em unsatisfied}; otherwise, a bidder
  is {\em satisfied}.  At termination of $\palloc(\alpha, \rho, \eps)$, all satisfied
  bidders are matched to a bundle $g(i)$ that is an $\alpha \cdot |g(i)|$-most
  preferred bundle.
\end{lemma}
\iffull
\begin{proof}
  We first show that a bidder's bundle $g(i)$ remains a subset of their most
  preferred bundle at the effective prices, i.e., with prices of goods in $g(i)$
  set to their price at time of assignment, and all other goods taking current
  prices.

  This claim follows by induction on the number of timesteps (ranging from $1$ to
  $nT$). The base case is clear. Now, assume that the claim holds up to time
  $t$.  There are three possible cases:
  \begin{enumerate}
    \item If the price of a good outside $g(i)$ is increased, $g(i)$ remains
      part of a most-preferred bundle by the gross substitutes condition.
    \item If the price of a good in $g(i)$ is increased, some goods may be
      removed from the bundle leading to a new bundle $g'(i)$. The only goods
      that experience an effective price increase lie outside of $g'(i)$, so
      $g'(i)$ remains a subset of a most-preferred bundle at the effective
      prices.
    \item If a bidder adds to their bundle, $g(i)$ is a subset of the
      most-preferred bundle by definition.
  \end{enumerate}
  Hence, a bidder becomes satisfied precisely when $g(i)$ is equal to the
  most-preferred bundle at the effective prices. The true price is at most
  $\alpha$ more than the effective price, so the bidder must have an $\alpha
  |g(i)|$-most preferred bundle at the true prices.
\end{proof}
\fi

\begin{lemma} \label{gs-unsatisfied}
  Suppose all counters have error at most $E$ throughout the run of
  $\palloc(\alpha, \rho, \eps)$. Then at termination, all but $\rho d$  bidders
  are satisfied if
  \[
    n \leq d
    \quad \text{and} \quad
    d \geq \frac{8E}{\rho} = \Omega \left( \frac{1}{\alpha \rho^2  \epsilon} \cdot
      \polylog \left( n, k, \frac{1}{\alpha}, \frac{1}{\rho},
        \frac{1}{\gamma} \right)  \right).
  \]
\end{lemma}
\iffull
\begin{proof}
  Note that as long as the algorithm does not halt, at least $\rho d - 4E$
  bidders are unsatisfied at the beginning of the round.  They may not actually
  bid when their turn comes, because the prices may have changed. Let the number
  of bids among all bidders be $B$, and suppose we run for $R$ rounds. We expect
  at least $\rho d - 4E$ bids per round, so $R(\rho d - 4E) - B$ is a lower
  bound on the number of times a bidder is unsatisfied but fails to bid.

  In the matching case, if a bidder is unsatisfied at the beginning of the round
  but fails to bid during their turn, this must be because the prices have risen
  too high. Since prices are monotonic increasing, such a bidder will never be
  unsatisfied again.

  In contrast, the gross substitutes case is slightly more subtle. Bidders who
  are unsatisfied at the beginning of a round and don't bid on their turn may
  later become unsatisfied again. Clearly, this happens only when the bidder
  loses at least one good after they decline to bid: if they don't lose any
  goods, then the prices can only increase after they decline to bid. Thus, they
  will have no inclination to bid in the future.

  There are at most $n$ cases of the bidder dropping out entirely. Thus, the
  number of times bidders report wanting to reenter the bidding is at least $R
  (\rho d - 4E) - n - B$. Since a bidder loses at least one good each time they
  reenter, the number of reentries is at most the number of bids $B$. Hence,
  the number of bids in $R$ rounds is at least
  \begin{equation} \label{eq:gs-bid}
    B \geq \frac{ R (\rho d - 4E) - n }{2}.
  \end{equation}

  Now, let $s' = s - m = s - (2E + 1)$ be the effective supply and consider how
  many bids are possible.  Each of the $k$ types of goods will accept at most
  $s' + 2E = s + 1$ bids at each of $1/\alpha$ price levels, so there are at
  most $k(s + 1)/\alpha = (d + k)/\alpha$ possible bids.
  Setting the left side of \cref{eq:gs-bid} equal to $(d + k)/\alpha$, we find
  \[
    R \leq \frac{1}{\alpha} \left( \frac{2(d + k) + \alpha n}{\rho d - 4E}
    \right) := T_0,
  \]
  so taking $T \geq T_0$ suffices to ensure that the algorithm halts with no
  more than $\rho d$ bidders unsatisfied. Assuming $\rho d \geq 8E$ and $d \geq
  n$,
  \begin{mathdisplayfull}
    T_0 \leq \fullfrac{10d}{\alpha \rho d} = \fullfrac{10}{\alpha \rho} = T.
  \end{mathdisplayfull}%
  The requirement on $n$ and $d$ is then
  \begin{equation*}
    d \geq \frac{8E}{\rho} = \Omega \left( \frac{1}{\alpha \rho^2  \epsilon} \cdot
      \polylog \left( n, k, \frac{1}{\alpha}, \frac{1}{\rho},
        \frac{1}{\gamma} \right)  \right)
    \quad \text{and} \quad
    n \leq d,
  \end{equation*}
  as desired.
\end{proof}
\fi

\begin{lemma} \label{gs-acc}
  With probability at least $1 - \gamma$, $\palloc(\alpha, \rho, \eps)$ computes an
  $(\alpha, \beta, \rho)$-approximate allocation equilibrium where
  \[
    \beta = O \left( \frac{1}{\alpha \rho \epsilon } \cdot \polylog
      \left(n, k,
        \frac{1}{\alpha}, \frac{1}{\rho}, \frac{1}{\gamma} \right) \right),
  \]
  so long as
  \[
    d \geq \frac{8E}{\rho} = \Omega \left( \frac{1}{\alpha \rho^2  \epsilon} \cdot
      \polylog \left( n, k, \frac{1}{\alpha}, \frac{1}{\rho},
        \frac{1}{\gamma} \right)  \right)  \text{and }
    n \leq d.
  \]
\end{lemma}
\iffull
\begin{proof}
  Condition on the error for each counter being at most $E$ throughout the run
  of the algorithm. By \cref{counter-error}, this holds for any single counter
  with probability at least $1 - \gamma/2k$. By a union bound, this holds for
  all counters with probability at least $1 - \gamma$. The theorem follows by
  \cref{gs-supply,gs-alpha,gs-unsatisfied}.
\end{proof}
\fi

Now, it is straightforward to prove the welfare theorem (\cref{gs-welfare}).
\iffull
\begin{proof}
  The proof follows the matching case (\cref{matching-welfare}) closely.  By
  \cref{gs-acc}, $(g,p)$ is a $(\alpha/3, \beta,
  \alpha/3)$-approximate allocation equilibrium, where $\beta = 4E' + 1$. Then  all but
  $\alpha d/3$ bidders are satisfied and get a bundle $g(i)$ that is $\alpha
  |g(i)|$ optimal; let this set of bidders be $B$. Note that $\sum_i |g(i)| \leq
  d$. Let $g^*$ be any other allocation. Then,
  \begin{align*}
    \sum_{i \in B} v_i (g(i)) - p(g(i)) &\geq \sum_{i \in B} v_i(g^*(i)) -
    p(g^*(i)) - \frac{\alpha }{3} |g(i)| \\
    \sum_{i \in B} v_i (g^*(i)) - v_i(g(i)) &\leq \sum_{i\in B} p(g^*(i)) -
    p(g(i)) + \alpha d/3
    = \sum_{j \in G} p_j (N^*_j - N_j) + \alpha d/3
  \end{align*}
  where the $N_j$ is the number of good $j$ sold in $g$ and $N_j^*$ is the
  number of good $j$ sold in $g^*$. If $p_j > 0$, we know $N_j \geq s - \beta$,
  hence $N_j^* - N_j \leq \beta \leq \alpha s/3$. Since $p_j\leq 1$ for
  each good $j$, we have
  \[
  \sum_{j\in G} p_j(N_j^* - N_j) \leq \sum_j  p_j(N_j^* - N_j) \leq
  \alpha \sum_j s = \alpha d/3.
  \]
  Furthermore, at most $\alpha d/3$ bidders are left unsatisfied in
  the end; these bidders contribute at most $\alpha d /3$ welfare to the optimal
  matching since valuations are bounded by $1$. Putting it all together,
  \[
    \sum_{i} v_i (g^*(i)) - v_i(g(i)) \leq \alpha d/3 + \alpha d/3 +
    \alpha d/3 = \alpha d.
  \]
  The stated supply bound $s$ follows directly from \cref{gs-acc}.
\end{proof}
\else
The proof follows the matching case (\cref{matching-welfare}) quite closely; we
defer the proof to the full version.
\fi

\section{Lower Bounds}
\label{sec:lowerbounds}

Our lower bounds all reduce to a basic database reconstruction lower bound for
differential privacy.
\begin{restatable}{theorem}{reconstructionbeta} \label{thm:reconstruction}
  Let mechanism $\cM \colon \{0,1\}^n \rightarrow \{0,1\}^n$ be
  $(\epsilon, \delta)$-differentially private, and suppose that for all
  database $D$, with probability at least $1-\beta$, $\|\cM (D) -
  D \|_1 \leq \alpha n$. Then,
  \[
    \alpha \geq 1 - \frac{e^\epsilon + \delta}{(1+e^\epsilon) (1-\beta)}
    := \theta(\epsilon, \delta, \beta).
  \]
\end{restatable}

In other words, no $(\epsilon, \delta)$-private mechanism can reconstruct more
than a fixed constant fraction of its input database.  For $\epsilon, \delta,
\beta$ small, $\theta(\epsilon, \delta, \beta) \sim 1/2$.  Informally, this theorem
states that a private reconstruction mechanism can't do much better than
guessing a random database.  Note that this holds even if the adversary doesn't
know which fraction was correctly reconstructed. This theorem is folklore; a
proof can be found in \thelongref{recons-details}.

Our lower bounds will all be proved using the following pattern.
\begin{itemize}
  \item First, we describe how to convert a database $D \in \{0, 1\}^n$ to
    a market, by specifying the bidders, the goods, and the valuations $v_{ij}
    \in [0, 1]$ on goods.
  \item Next, we analyze how these valuations change when a single bit in $D$ is
    changed. This will control how private the matching algorithm is with
    respect to the original database, when applied to this market.
  \item Finally, we show how to output a database guess $\hat{D}$ from the
    matching produced by the private matching algorithm.
\end{itemize}
This composition of three steps will be a private function from $\{0, 1\}^n
\rightarrow \{0, 1\}^n$, so we can apply \cref{thm:reconstruction} to lower
bound the error, implying a lower bound on the error of the matching algorithm.

\subsection{Standard Differential Privacy}
Note that \cref{alg:matching} produces market clearing prices under standard
differential privacy. We will first show that this is not possible if each good
has unit supply. Recall that prices correspond to an {\em $(\alpha, \beta,
  \rho)$-approximate matching equilibrium} if all but $\rho$ bidders can be
allocated to a good such that their utility is within
$\alpha$ of their favorite good (\cref{matching-eq}). We will ignore the $\beta$
parameter, which controls how many goods are left unsold.
\iffull\else (We defer the proof to the full version.)\fi

% \ar{We should use the same notion we already defined earlier as
%   $(\alpha,\beta,\rho)$-approximate matching equilibrium}.
% \jh{Fixed.}

\begin{theorem} \label{lb-prices}
  Let $n$ bidders have valuations $v_{ij} \in [0, 1]$ for $n$ goods. Suppose
  that mechanism $\cM$   is $(\epsilon, \delta)$-differentially private, and
  calculates prices corresponding to an $(\alpha, \beta, \rho)$-approximate
  matching equilibrium for $\alpha < 1/2$ and some $\beta$ with probability
  $1 - \gamma$. Then,
  \begin{mathdisplayfull}
    \rho \geq \frac{1}{2} \theta(2 \epsilon, \delta(1 + e^\epsilon), \gamma).
  \end{mathdisplayfull}%
  Note that this is independent of $\alpha$.
\end{theorem}
\iffull
\begin{proof}
  Let $D \in \{0, 1\}^{n/2}$ be a private database and construct the following
  market. For each bit $i$ we construct the following gadget, consisting of two
  goods $\mathbf{0}_i, \mathbf{1}_i$ and two bidders, $b_i, \bar{b_i}$. Both
  bidders have valuation $D_i$ for good $\mathbf{1}_i$, $1 - D_i$ for good
  $\mathbf{0}_i$, and valuation $0$ for the other goods. Evidently, there are
  $n$ bidders and $n$ goods.

  Note that changing a bit $i$ in $D$ changes the valuation of exactly two
  bidders in the market: $b_i$ and $\bar{b_i}$. Therefore, mechanism $\cM$  is
  $(2\epsilon, \delta(1 + e^\epsilon))$-differentially private with respect to
  $D$. Let the prices be $p_{0i}, p_{1i}$. To guess the database $\hat{D}$, we
  let $\hat{D}_i = 1$ if $p_{1i} > 1/2$, otherwise $\hat{D}_i = 0$.

  By assumption, $\cM$   produces prices corresponding to an $(\alpha, \beta,
  \rho)$-approximate matching equilibrium with probability $1 - \gamma$.  We do
  not have access to the matching, but we know the prices must correspond to
  {\em some} matching $\mu$. Then, for all but $\rho n$ gadgets, $\mu$ matches
  both bidders to their $\alpha$-approximate favorite good and both goods are matched
  to bidders who receive $\alpha$-approximate favorite goods.

  Consider such a gadget $i$. We will show that exactly one of
  $p_{0i}$ or $p_{1i}$ is greater than $1/2$, and this expensive good
  corresponds to bit $D_i$. Consider one of the bidders in this gadget, and
  suppose she prefers good $g_+$ with price $p_+$, while he received good $g_-$
  with price $p_-$.  Since she receives an $\alpha$-approximate favorite good,
  \begin{mathdisplayfull}
    (1 - p_+) - (0 - p_-) \leq \alpha,
    \longquad \text{so} \longquad
    p_+ - p_-             \geq 1 - \alpha > 1/2.
  \end{mathdisplayfull}%
  So $p_+ > 1/2$ and $p_- < 1/2$. Note that good $g_+$ is in the gadget, while
  good $g_-$  may not be. So, one of the goods in the gadget has price strictly
  greater than $1/2$. The other good in the gadget is an
  $\alpha$-approximate favorite
  good for some bidder. All bidders have valuation $0$ for the good, hence its
  price must be strictly less than $1/2$.

  Thus, the reconstruction procedure will correctly produce bit for each such
  gadget, and so will miss at most $\rho n$ bits with probability at least $1 -
  \gamma$. The combined reconstruction algorithm is a map from $\{0, 1\}^{n/2}
  \rightarrow \{0, 1\}^{n/2}$, and $(2\epsilon, \delta(1 +
  e^\epsilon))$-differentially private. By \cref{thm:reconstruction},
  \begin{mathdisplayfull}
    2 \rho \geq \theta(2 \epsilon, \delta(1 + e^\epsilon), \gamma).
  \end{mathdisplayfull}%
\end{proof}
\fi

\subsection{Separation Between Standard and Joint Differential Privacy}

While we can compute an approximate maximum-weight matching under joint privacy
when the supply of each good is large \shortbreak (\cref{matching-acc}), this is not
possible under standard differential privacy even with infinite supply.
% (In fact, it is not possible with finite supply either.)

\begin{theorem} \label{lb-alloc}
  Let $n$ bidders have valuations $v_{ij} \in \{0, 1\}$ for $2$ goods with
  infinite supply. Suppose that mechanism $\cM$  is $(\epsilon,
  \delta)$-differentially private, and computes a matching with weight at least
  $\OPT - \alpha n$ with probability $1 - \gamma$. Then,
  \begin{mathdisplayfull}
    \alpha \geq \theta(\epsilon, \delta, \gamma).
  \end{mathdisplayfull}%
  % \ar{Does the other lower bound for standard differential privacy show anything
  %   this one doesn't? Why not just use this one?}
  % \jh{Yes, this one is better.}
\end{theorem}
\begin{proof}
  Let $D \in \{0, 1\}^{n}$.  We assume two goods, $\mathbf{0}$ and $\mathbf{1}$.
  We have one bidder $b_i$ for each bit $i \in [n]$, who has valuation $D_i$ for
  $\mathbf{1}$, and valuation $1 - D_i$ for $\mathbf{0}$.  Since changing a bit
  changes a single bidder's valuation, applying $\cM$  to this market is
  $(\epsilon, \delta)$-private with respect to $D$. To guess the database
  $\hat{D}$, we let $\hat{D}_i$ be $0$ if $b_i$ is matched to $\mathbf{0}$, $1$
  if $b_i$ is matched to $\mathbf{1}$, and arbitrary otherwise.

  Note that the maximum welfare matching assigns each $b_i$ the good
  corresponding to $D_i$, and achieves social welfare $\OPT = n$. If $\cM$ computes an
  matching with welfare $\OPT - \alpha n$, it must give all but an
  $\alpha$ fraction of bidders $b_i$ the good corresponding to $D_i$.  So, the
  reconstructed database will miss at most $\alpha n$ bits with probability $1 -
  \gamma$, and by \cref{thm:reconstruction},
  \begin{mathdisplayfull}
    \alpha \geq \theta(\epsilon, \delta, \gamma).
  \end{mathdisplayfull}%
\end{proof}
Note that this gives a separation: under joint differential privacy,
\cref{alg:matching} can release a matching with welfare $\OPT - \alpha n$ for
any $\alpha$, provided supply $s$ is large enough (by \cref{matching-welfare}),
while this is not possible under standard differential privacy even with {\em
  infinite} supply.

\subsection{Joint Differential Privacy}

Finally, we show that a large supply assumption is necessary in order to compute
an additive $\alpha$ maximum welfare matching under joint differential privacy.

\begin{theorem} \label{lb-jp}
  Let $n$ bidders have valuations $v_{ij} \in [0, 1]$ for $k$ types of goods
  with supply $s$ each. Suppose mechanism $\cM$  is $(\epsilon, \delta)$-joint
  differentially private for $\epsilon, \delta < 0.1$, and calculates a matching
  with welfare at least $\OPT - \alpha n$ with probability $1 - \gamma$ for
  $\gamma < 0.01$, and all $n, k, s$.  Then, $s = \Omega(\sqrt{1/\alpha}).$
\end{theorem}
% \ar{What do we mean by an $\alpha$-approximate max welfare matching here?
%   Normally we seem to mean error $\alpha\cdot n$, but by the scaling here, it
%   seems we mean additive error $\alpha$. Lets be consistent and say that to get
%   a matching of weight $\geq$ OPT - $\alpha n$, we need $s \geq BLAH$. }
% \jh{Fixed.}
\iffull
\begin{proof}
  Let $k = n/(s+1)$.
  % \sw{$k$ is different from the theorem statement?}
  % \jh{My mistake.}
  Given a private database $D \in \{0, 1\}^k$, construct the following market.
  For each bit $i$, we construct a gadget with two goods $\mathbf{0}_i,
  \mathbf{1}_i$, each with supply $s$. Each gadget has a distinguished bidder
  $b_i$ and $s$ identical bidders, all labeled $\bar{b_i}$.  Let bidder $b_i$,
  who we call the {\em real bidder}, have valuation $D_i$ for $1_i$, and $1 -
  D_i$ for $0_i$.  Bidders $\bar{b_i}$, which we call the {\em spy bidders}, all
  have the same valuation: $\eta = \frac{1}{4s}$
  % \sw{$4/s$ ?}
  % \jh{Fixed.}
  for $\mathbf{0}_i$ or $\mathbf{1}_i$ drawn at random, and $0$ for all other
  goods (in and out of the gadget). We say a bidder {\em prefers} a good if they
  have positive valuation for the good.
  % \sw{should we say both the real and spy bidders only prefer goods in
  % the same gadget?}
  % \jh{Done.}

  Note that changing a bit in $D$ changes a single bidder's valuation. Also note
  that the spy bidders' valuations do not depend on $D$. Hence, by joint
  differential privacy of $\cM$, the function that maps the above market
  through $\cM$  to the allocation of just the spy bidders is $(\epsilon,
  \delta)$-differentially private with respect to an entry change in $D$.

  We will describe how to guess $\hat{D}$ based on just the spy bidders' joint
  view, i.e., the goods they are assigned. This reconstruction procedure will
  then be $(\epsilon, \delta)$-differentially private, and we can apply
  \cref{thm:reconstruction} to lower bound the error of $\cM$ . For every bit $i
  \in [k]$,
  % \sw{$[k]$?}
  % \jh{Yes, thanks.}
  let $\hat{D}_i$ be $1$ if the spy bidders in gadget $i$ are all assigned to
  $\mathbf{0}_i$, $0$ if the spy bidders in gadget $i$ are all assigned to
  $\mathbf{1}_i$, and uniformly random otherwise.

  We'll say that a gadget {\em agrees} if the spy bidders and real bidder prefer
  the same good. Gadgets that don't agree, {\em disagree}.  Let $w$ be the
  number of gadgets that agree.
%  \sw{$s$ already shows up as the supply}
  By construction, gadgets agree independently with probability $1/2$ each.
  Hence, Hoeffding's inequality gives
  \[
    \prob{\left|w - \frac{k}{2}\right| \leq \lambda k} \geq 1 -
    2\exp(-2 \lambda^2 k)
  \]
  for some $\lambda$ to be specified later; condition on this event. With
  probability at least $1 - \gamma$, mechanism $\cM$ computes a matching with
  welfare at least $\OPT - \alpha n$; condition on this event as well.  Note
  that the optimum welfare is $1 + (s - 1) \eta$ for gadgets that agree, and $1
  + s \eta$ for gadgets that disagree, hence $\OPT = w  (1 + (s-1) \eta) + (k -
  w) (1 + s \eta)$ in total.
  % \sw{might have mixed two $s$'s together}
  % \jh{Yes, fixed now hopefully.}

%  \ar{This argument is about individual gadgets. Say that.}

  For each gadget, there are several possible allocations.  Intuitively, an
  assignment gives social welfare, but may also lead to a bit being
  reconstructed. Let $RB(\mu) = \|D - \hat{D}\|_1$ be the error of the
  reconstruction when the matching is $\mu$. We'll argue that any matching $\mu$
  with nearly optimal social welfare must result in large expected
  reconstruction $\mathbb{E}[RB(\mu)]$. By linearity,
  \[
    \mathbb{E}[RB(\mu)] = \sum_{i \in [k]} \prob{D_i = \hat{D}_i},
  \]
  so it suffices to focus on gadget at a time.

  First, suppose the gadget $i$ agrees. The matching $\mu$ can give the
  preferred good to the bidder, the spies, or neither. If the preferred good
  goes to the bidder, this gives at most $1 + (s -1) \eta$ social welfare. Not
  all the spies get the same good, so
  \begin{mathdisplayfull}
    \prob{D_i = \hat{D}_i} = \fullfrac{1}{2}.
  \end{mathdisplayfull}%
  If the preferred good goes to the spies, then this contributes $s \eta$ to
  social welfare, and
  \begin{mathdisplayfull}
    \prob{D_i = \hat{D}_i} = 0.
  \end{mathdisplayfull}%
  Note that it doesn't matter whether the bidder is assigned in $\mu$, since the
  social welfare is unchanged and the reconstruction algorithm doesn't have
  access to the bidder's allocation. There are other possible allocations, but
  they are dominated by these two choices since they get less social welfare for
  higher reconstruction probability.

  Now, suppose gadget $i$ disagrees. There are several possible allocations.
  First, both the bidder and the spies may get their favorite good.  This gives
  $1 + s \eta$ welfare, and
  \begin{mathdisplayfull}
    \prob{D_i = \hat{D}_i} = 1.
  \end{mathdisplayfull}%
  Second, the bidder may be assigned their favorite good, and at most $s -1$
  spies may be assigned their favorite good.  This leads to $1 + (s - 1) \eta$
  welfare, with
  \begin{mathdisplayfull}
    \prob{D_i = \hat{D}_i} = \fullfrac{1}{2}.
  \end{mathdisplayfull}%
  Again, there are other possible allocations, but they lead to less social
  welfare or higher reconstruction probability. We say the four allocations
  above are {\em optimal}.

  Let $a_1, a_2$ be the fractions of agreeing gadgets with the two optimal
  agreeing allocations, and $d_1, d_2$ be the fractions of disagreeing gadgets
  with the two optimal disagreeing allocations.  Let $t$ be the fraction of
  agreeing pairs. The following linear program minimizes $(1/k)
  \mathbb{E}[RB(\mu)]$ over all matchings $\mu$ achieving an
  $\alpha$-approximate maximum welfare matching for supply $s$.
  \begin{align*}
    LP_s         & := & \text{minimize: }  & \frac{1}{2}a_1 + d_1 + \frac{1}{2}d_2 \\
                 &    & \text{such that: } & a_1 + a_2 \leq t \\
                 &    &                    & d_1 + d_2 \leq 1 - t \\
                 &    &                    & \frac{1}{2} - \lambda \leq t \leq \frac{1}{2} + \lambda \\
                 &    &                    & (1 + (s - 1) \eta) a_1 + s \eta a_2
                                            + (1 + s \eta) d_1 + ( 1+ (s - 1) \eta) d_2 \\
                 &    &                    & \geq t (1 + (s - 1) \eta) + (1 - t)
                 (1 + s \eta) - \frac{\alpha n}{k}
  \end{align*}
  The last constraint is the welfare requirement, the second to last constraint
  is from conditioning on the number of agreeing gadgets, and the objective is
  $(1/k) \mathbb{E}[RB(\mu)]$.

  Plugging in $\eta = \frac{1}{4s}, \lambda = 1/128, \alpha = \frac{k}{16ns}$
  and solving, we find
  \[
    (a_1, a_2, d_1, d_2, t) = \left(\frac{65}{128}, 0, \frac{31}{128},
      \frac{1}{4}, \frac{65}{128}\right)
  \]
  is a feasible solution for all $s$ with objective $\alpha' = 159/256$. To
  show that this is optimal, consider the dual problem:

  \begin{align*}
    DUAL_s          & := & \text{maximize: }  &
                                - \rho_2
                                + \left( \frac{1}{2} - \lambda \right) \rho_3
                                - \left( \frac{1}{2} + \lambda \right) \rho_4
                                + \left( 1 + s \eta - \frac{\alpha n}{k}\right) \rho_5 \\
                    &    & \text{such that: } & - \rho_1 + (1 + (s - 1) \eta)
                                                \rho_5 \leq \frac{1}{2} \\
                    &    &                    & - \rho_1 + s \eta \rho_5 \leq 0 \\
                    &    &                    & - \rho_2 + (1 + s \eta) \rho_5 \leq 1 \\
                    &    &                    & - \rho_2 + (1 + (s -1) \eta)
                                                \rho_5 \leq \frac{1}{2} \\
                    &    &                    & \rho_1 - \rho_2 + \rho_3 -
                                                \rho_4 + \eta \rho_5 \leq 0
  \end{align*}
  We can directly verify that
  \[
    (\rho_1, \rho_2, \rho_3, \rho_4, \rho_5) = \left( \frac{5}{2} s - 1,
    \frac{5}{2}s - 1, 0, \frac{1}{2}, 2s \right)
  \]
  is a dual feasible solution with objective $\alpha' = 159/256$.

  We know that $\cM$ calculates an additive $\alpha$-approximate maximum welfare
  matching. While the allocations to each gadget may not be an optimal
  allocation, suboptimal allocations all have less social welfare and larger
  $RB$. So, we know the objective of $LP_m$ is a lower bound for $RB(\cM)$.

  Thus, $\mathbb{E}[RB(\cM)] \geq k \alpha'$ for any supply $s$. Since $RB$ is
  the sum of $k$ independent, $0/1$ random variables, another Hoeffding bound
  yields
  \begin{mathdisplayfull}
    \prob{ RB(\cM)/k \geq \alpha' - \lambda' } \geq 1 - 2 \exp (-2\lambda'^2 k).
  \end{mathdisplayfull}%
  Set $\lambda' = 1/256$, and condition on this event. All
  together, any matching mechanism $\cM$ which finds a matching with weight at
  least $\OPT - \alpha n$ failing with at most $\gamma$ probability gives an
  $(\epsilon, \delta)$-private mechanism mapping $D$ to $\hat{D}$ such
  that
  \begin{mathdisplayfull}
    \frac{1}{k} \cdot \|D - \hat{D}\|_1 \geq \alpha' - \lambda' = 79/128.
  \end{mathdisplayfull}%
  with probability at least $1 - \gamma - 2\exp(-2 \lambda^2 k) - 2\exp(-2
  \lambda'^2 k)$.

  For $\epsilon, \delta < 0.1$ and $\gamma < 0.01$, this contradicts
  \cref{thm:reconstruction} for large $k$.  Note that the failure probability
  and accuracy do not depend directly on $s$ since $\lambda, \lambda', \alpha'$
  are constants.  Hence
  \[
    \alpha \gg \frac{k}{16ns} = \frac{1}{16s(s + 1)}
  \]
  uniformly for all $s$, and $s = \Omega(\sqrt{1/\alpha})$ as desired.
\end{proof}
\else
We will only sketch the idea here, deferring the full proof to the full version.
Given a database $D \in \{0, 1\}^n$, we will have one real bidder, $m$ ``spy''
bidders, and two goods for each bit. The real bidder will have valuation for one
of the two goods determined by the private data $D$, while the spy bidders will
all have the same preference for one of the two goods, set uniformly at random
(independent of the private data). By arranging the valuations of the spy
bidders appropriately, we can show that any algorithm that achieves good welfare
must serve many of the spy bidders. When the spy bidder and the true bidder
prefer the same good (which happens half of the time), the spy bidders can learn
about the true bidder's preferences when they don't get their preferred good. By
taking the joint view of spy bidders, we can reconstruct a large enough portion
of the database to contradict \cref{thm:reconstruction} since under {\em joint}
differential privacy, the view of the spy bidders should satisfy {\em standard}
differential privacy with respect to the data from outside the coalition, i.e.,
the private data.
\fi

\section{Conclusion and Open Problems}

In this paper we gave algorithms to accurately solve the private
allocation problem when bidders have gross substitute valuations, achieving
joint differential privacy when the supply of each good is growing at least
logarithmically in the number of agents. Our results
are qualitatively tight: it is not possible to strengthen our approach
to standard differential privacy (from joint differential privacy), nor is it
possible to solve even max-matching problems to non-trivial accuracy under joint
differential privacy with constant supply.  Moreover, it is not clear how to
extend our approach to more general valuations: our algorithm fundamentally
relies on computing Walrasian equilibrium prices for the underlying market, and
such prices are not guaranteed to exist for valuation functions beyond the gross
substitutes class. This does not mean that the allocation problem cannot be
solved for more general valuation functions---rather, new ideas
seem to be needed.

Along with \citet{kearns-largegame} and other works in the joint privacy model,
our work adds compelling evidence that substantially more is possible under the
relaxation of \emph{joint} differential privacy compared to the standard
notion of differential privacy. For both the allocation problem studied here, 
and the equilibrium computation problem studied in \citet{kearns-largegame},
non-trivial results are impossible under differential privacy while strong
results can be derived under joint-differential privacy. Characterizing the
power of joint differential privacy, compared to standard differential
privacy, is a fascinating direction for future work.

More specifically, in this paper we achieved joint differential privacy via the
{\em billboard lemma}: we showed that the allocation given to each player can be
derived as a deterministic function only of 1) a differentially private message
revealed to all players, and 2) their own private data. However, this isn't
necessarily the only way to achieve joint differential privacy. How much further
does the power of joint differential privacy extend beyond the billboard model?

\iffull
\appendix
\section{Privacy Analysis for Counters}
\label{counter-details}

\citet{chan-counter} show that $\counter(\epsilon, T)$ is
$\epsilon$-differentially private with respect to single changes in the input
stream, when the stream is generated non-adaptively. For our application we
require privacy to hold for a large number of streams whose joint-sensitivity
can nevertheless be bounded, and whose entries can be chosen adaptively.  To
show that \counter is also private in this setting (when $\epsilon$ is set
appropriately), we first introduce some differential privacy notions.

We will make use of a basic differentially private mechanism originally due to
\cite{DMNS06}.
\begin{theorem}[\cite{DMNS06}]
  For a function $f:\mathcal{D}\rightarrow \mathbb{R}$, let
\[
  \Delta_1 = \max_{D, D' \in \mathcal{D}}\frac{|f(D)-f(D')|}{|\{i : D_i \neq
    D'_i\}|}
\]
  denote the $\ell_1$ sensitivity of $f$. Then the \emph{Laplace mechanism}
  which on input $D$ outputs $f(D) + \textrm{Lap}(\Delta_1/\epsilon)$ is
  $\epsilon$-differentially private. Here, $\textrm{Lap}(b)$ denotes a random
  variable drawn from the Laplace distribution with parameter $b$.
\end{theorem}
% \jh{Why is this here? To prove privacy of the counter mechanism?}
% \ar{Ya, it seemed to me at various points in the privacy proof, we appeal to the
%   privacy properties of the Laplace mechanism. Correct me if not.}
% \jh{OK, that's fine. }
\subsection{Composition}
An important property of differential privacy is that it degrades gracefully
when private mechanisms are composed together, even adaptively. We recall the
definition of an adaptive composition experiment
\citep{dwork-composition}.

\begin{definition}[Adaptive composition experiment]
  \label{comp-exp}
\leavevmode
\begin{itemize}
  \item Fix a bit $b \in \{0, 1\}$ and a class of mechanisms $\cM$.
  \item For $t = 1 \dots T$:
    \begin{itemize}
      \item The adversary selects databases $D^{t, 0}, D^{t, 1}$ and a
        mechanism $\cM_t \in \cM$.
      \item The adversary receives $y_t = \cM_t(D^{t, b})$
    \end{itemize}
\end{itemize}
\end{definition}
The output of an adaptive composition experiment is the view of the
adversary over the course of the experiment. The experiment is said to be
$\epsilon$-differentially private if
\[
  \max_{S \subseteq \cR}\frac{\Pr[V^0 \in S]}{\Pr[V^1 \in S]} \leq
  \exp(\epsilon),
\]
where $V^0$ is the view of the adversary with $b = 0$, $V^1$ is the view of the
adversary with $b = 1$, and $\cR$ is the range of outputs.

Any algorithm that can be described as an instance of this adaptive composition
experiment for some adversary is said to be an instance of the class of
mechanisms $\cM$ under \emph{adaptive $T$-fold composition}. We now state a
straightforward consequence of a composition theorem by
\citet{dwork-composition}.

\begin{lemma}[\citet{dwork-composition}]
  \label{lem:l1-comp}
  Let $\Delta_1 \geq 0$. The class of $\frac{\epsilon}{\Delta_1}$-private
  mechanisms satisfies $\epsilon$-differential privacy under adaptive
  composition, if the adversary always selects databases satisfying
  \[
    \sum_{t = 1}^T \left|D^{t, 0} - D^{t, 1}\right| \leq \Delta_1.
  \]
\end{lemma}
In other words, the privacy parameter of each mechanism should be calibrated for
the total distance between the databases over the whole composition (the {\em
  $\ell_1$ sensitivity}).

\subsection{Binary mechanism}
We reproduce the Binary mechanism here in order to refer to its internal
workings in our privacy proof.
First, it is worth explaining the intuition of the \counter. Given a bit stream
$\sigma \colon [T] \rightarrow \{0,1\}$, the algorithm releases the counts
$\sum_{i=1}^t \sigma(i)$ for each $t$ by maintaining a set of partial sums
$\sum[i, j] \coloneqq \sum_{t=i}^j \sigma(t)$. More precisely, each partial sum
has the form $\Sigma[2^i + 1, 2^i + 2^{i - 1}]$, corresponding to powers of $2$.

In this way, we can calculate the count $\sum_{i=1}^t \sigma(i)$ by summing at
most $\log{t}$ partial sums: let $i_1 < i_2 \ldots < i_m$ be the indices of
non-zero bits in the binary representation of $t$, so that
\begin{equation*}
  \label{eq:binary}
  \sum_{i=1}^t  \sigma(i) = \sum[1, 2^{i_m}] + \sum[2^{i_m}+1, 2^{i_m}
  + 2^{i_{m-1}}] + \ldots + \sum[t-2^{i_1} + 1, t].
\end{equation*}
Therefore, we can view the algorithm as releasing partial sums of
different ranges at each time step $t$ and computing the counts is
simply a post-processing of the partial sums. The core algorithm is
presented in \cref{alg:binary-mechanism}.

% \ar{For presentation, perhaps put the description of the algorithm in the
%   appendix and just state the theorem. Also, do we ever use the utility theorem
%   here? We only deal with counters that have sensitivity > 1. If we don't use
%   it, no need to quote it. Finally, I know they use this term, but is "streaming
%   privacy" any different than "differential privacy", where the output is the
%   stream? We don't want to seem like we have a third definition of privacy
%   here...}
% \sw{We do use utility theorem. Later we plug in privacy parameter $\frac{\epsilon}{T}$}

\begin{algorithm}[h!]
     \caption{$\counter(\eps, T)$}
     \begin{algorithmic}\label{alg:binary-mechanism}
       \STATE{\textbf{Input:} A stream $\sigma\in \{0,1\}^T$}
       \STATE{\textbf{Output: } $B(t)$ as estimate for $\sum_{i=1}^t
         \sigma(i)$ for each time $t\in [T]$}
       \FORALL{$t\in [T]$}
       \STATE Express $\displaystyle t = \sum_{j=0}^{\log{t}} 2^j\bin_j(t)$.
       \STATE Let $i \leftarrow \min_j\{\bin_j(t) \neq 0\}$
       \STATE $a_i\leftarrow \sum_{j < i} a_j + \sigma(t) $,
               $(a_i  = \sum[t-2^i + 1, t])$
               \FOR{$0\leq j \leq i - 1$}
               \STATE Let $a_j \leftarrow 0$ and $\hat{a_j} \leftarrow   0$
     \ENDFOR
     \STATE Let $\hat{a_j} = a_j + \Lap(\log(T) /\eps)$
     \STATE Let $\displaystyle B(t) = \sum_{i: \bin_i(t)\neq 0} \hat{a_i}$
      \ENDFOR

     \end{algorithmic}
\end{algorithm}

\subsection{Counter Privacy Under Adaptive Composition}
We can now show that the prices released by our mechanism
satisfy $\epsilon$-differential privacy.

\counterpriv*
\begin{proof}
  \citet{chan-counter} show this for a single sensitivity 1 counter for a
  non-adaptively chosen stream. We here show the generalization to multiple
  counters run on adaptively chosen streams with bounded $\ell_1$ sensitivity,
  and bound the $\ell_1$ sensitivity of the set of streams produced by our
  algorithm.  We will actually show that the sequence of noisy partial sums released by
  \counter satisfy $\epsilon$-differential privacy. This is only stronger: the
  running counts are computed as a function of these noisy partial sums.

  To do so, we first define an adversary for the adaptive composition
  experiment (\cref{comp-exp}) and then show that the view of this adversary is
  precisely the sequence of noisy partial sums. The composition theorem
  (\cref{lem:l1-comp}) will then show that the sequence of noisy partial sums
  are differentially private with respect to a change in a bidder's valuation.

  Let the two runs $b = 0, 1$ correspond to any two neighboring valuations
  $(v_i, v_{-i})$ and $(v'_i, v_{-i})$ that differ only in bidder $i$'s
  valuation. We first analyze the view on all of the
  $\text{counter}(j)$ for $j = 1,\ldots, k$.

  The adversary will operate in phases. There are two kinds of phases, which we
  label $P_t$ and $P'_t$: one phase per step of the good counters, and one phase
  per step of the halting condition counter. Both counters run from time $1$ to
  $nT$, so there are $2nT$ phases in total.

  At each point in time, the adversary maintains histories $ \{b_i\}, \{b'_i\}$
  of all the bids prior to the current phase and histories $\{e_i\}, \{e'_i\}$
  of all prior reports to the halting counter $\text{counter}_0$, when bidder $i$
  has valuation $v_i, v'_i$ respectively.

  Let us consider the first kind of phase. One bidder bids per step of the
  counter, so one bidder bids in each of these phases.  Each step of the
  experiment the adversary will observe a partial sum.  Suppose the adversary is
  in phase $P_t$.  Having observed the previous partial sums, the adversary can
  simulate the action of the current bidder $q$ from the histories of previous
  bids by first computing the prices indicated by the previous partial sums. The
  adversary will compute $q$'s bid when the valuations are $(v_i, v_{-i})$, and
  when the valuations are $(v_i', v_{-i})$.  Call these two bids $b_t, b_t'$
  (which may be $\perp$ if $q$ is already matched in one or both of the
  histories).

  Note that for bidders $q \neq i$, it is always the case that $b_t = b_t'$.
  This holds by induction: it is clearly true when no one has bid, and bidder
  $q$'s decision depends only on her past bids, the prices, and her valuation.
  Since these are all independent of bidder $i$'s valuation, bidder $q$ behaves
  identically.

  After the adversary calculates $b_t, b_t'$, the adversary simulates update and
  release of the counters. More precisely, the adversary spends phase $P_t$
  requesting a set of partial sums
  \[
    \Sigma = \{ \sigma^j_{I} \mid j \in [k], I \in S_t \},
  \]
  where $S_t \subseteq [1, nT]$ is a set of intervals ending at $t$,
  corresponding to partial sums that \counter releases at step $t$.

  For each $\sigma^j_I \in \Sigma$, $D^0, D^1 \in \{0, 1\}_I$ are defined by
  \[
    D^0_k = \left\{
      \begin{array}{ll}
        1 &: \text{if } b_k = j \\
        0 &: \text{otherwise}
      \end{array}
    \right.
  \]
  and similarly for $D^1$, with bid history $\{b'_i\}$. Informally, a database
  $D$ for $\sigma^j_I$ encodes whether a bidder bid on good $j$ at every
  timestep in $I$.  The adversary will define $\cM$ to sum the bits in the
  database and add noise $Lap(1/\epsilon_0)$, an $\epsilon_0$-differentially
  private operation.  Once the partial sums for $P_t$ are released, the
  adversary advances to the next phase.

  Now, suppose the adversary is in the second kind of phase, say $P'_t$. This
  corresponds to a step of the halting condition counter. We use exactly the
  same construction as above: the adversary will request the partial sums
  corresponding to each timestep. The adversary will simulate each bidder's
  action by examining the history of bids and prices. Now suppose the two runs
  differ in bidder $i$'s valuation. Following the same analysis, the reports to
  this halting condition counter differ only in bidder $i$'s reports.

  With this definition, the view of the adversary on database $\{D^0\}$ and
  $\{D^1\}$ is precisely the noisy partial sums when the valuations are $(v_i,
  v_{-i})$ and $(v_i', v_{-i})$, respectively. So, it suffices to show that
  these views have almost the same probability.

  We apply \cref{lem:l1-comp} by bounding the distance between the databases for
  counter$(1)$ to counter$(k)$.  Note that the sequence of databases $ \{D^0\},
  \{D^1\}$ chosen correspond to streams of bids that differ only in bidder $i$'s
  bid, or streams of reports to $\text{counter}(0)$ that differ only in bidder
  $i$'s report. This is because the bid histories $\{b_t\}, \{b_t'\}$ and report
  histories $\{e_t\}, \{e_t'\}$ differ only on timesteps where $i$ acts.  Thus,
  it suffices to focus on bidder $i$ when bounding the distance between these
  databases.

  Consider a single good $j$, and suppose $c_j$ of $i$'s bids on good $j$ differ
  between the histories. Each of bidder $i$'s bids on $j$ show up in $\log (nT)$
  databases, so
  \[
    \sum |D^0_j - D^1_j| \leq c_j \log nT,
  \]
  where the sum is taken over all databases corresponding to good
  $j$. The same is true for the halting condition counter: if there are $c_0$
  reports that differ between the histories, then
  \[
    \sum |D^0_0 - D^1_0| \leq c_0 \log nT.
  \]

  Since we know that a bidder can bid at most $T$ times over $T$ proposing
  rounds, and will report at most $T$ times, we have $\ell_1$ sensitivity
  bounded by
  \[
    \Delta_1 \leq c_0 \log n T + \sum_j c_j \log nT \leq 2T \log nT.
  \]
  By \cref{lem:l1-comp}, setting
  \[
    \epsilon_0 = \frac{\epsilon}{2T\log nT}
  \]
  suffices for $\epsilon$-differential privacy, and this is precisely running
  each \counter with privacy level $\epsilon' = \epsilon/2T$.
\end{proof}

\section{Reconstruction Lower Bound} \label{recons-details}
Here, we detail a basic lower bound about differential privacy. Intuitively, it
is impossible for an adversary to recover a database better than random guessing
from observing the output of a private mechanism. The theorem is folklore.

% \sw{gets a different bound $1-\alpha \leq \frac{1+\delta}{(1+e^{-\epsilon})(1 -
%     \beta)}$ if we use a alternative DP expression}
\reconstructionbeta*
\begin{proof}
Fix a database $D\in \{0,1\}^n$ and sample an index $i$ uniformly at
random from $[n]$. Let $D'$ be a neighboring database of $D$ that
differs at the $i$-th bit. By assumption, we have that with
probability at least $1-\beta$
\[
  \| \cM(D) - D \|_1 \leq \alpha n,
  \qquad
  \| \cM(D') - D' \|_1 \leq \alpha n.
\]
Since $i$ is chosen uniformly, we then have
\[
  \Pr[ \cM(D)_i = D_i ] \geq (1 - \alpha)(1 - \beta),
  \qquad
  \Pr[ \cM(D')_i = D'_i ] \geq (1 - \alpha)(1 - \beta).
\]
It follows that $\Pr[ \cM(D')_i = D_i ]  \leq 1 - (1 - \alpha)(1- \beta)$
because $D_i \neq D_i'$.  By definition of $(\epsilon, \delta)$-differential
privacy, we get
\[
  (1 - \alpha)(1 - \beta)  \leq \Pr[\cM(D)_i = D_i ] \leq
  e^\epsilon \Pr[\cM(D')_i = D_i] + \delta \leq e^\epsilon (1 - (1
  - \alpha)(1 - \beta)) + \delta.
\]
Then we have
\[
  1 - \alpha \leq \frac{e^\epsilon + \delta}{(1+ e^\epsilon)(1 - \beta)}
\]
as desired.
\end{proof}

\fi

\subsection*{Acknowledgments}
The authors would like to thank Cynthia Dwork, Sudipto Guha, Moritz Hardt,
Sanjeev Khanna, Scott Kominers, Mallesh Pai, David Parkes, Adam Smith, and
Kunal Talwar for helpful discussions. In particular, we would like to thank
Scott Kominers for suggesting the connection to Kelso and Crawford, and Adam
Smith for discussions on the ``billboard model'' of privacy. Finally, we
thank the anonymous reviewers.

\iffull
\bibliographystyle{plainnat}
\else
% \footnotesize
\bibliographystyle{acmtrans}
\fi
\bibliography{./refs}

\end{document}